\documentclass[12pt,oneside]{article}

\usepackage{latexsym,amssymb,amsfonts,amsmath,amsthm,ifthen}

\usepackage{mathrsfs}

\allowdisplaybreaks[3]

\theoremstyle{plain}
\newtheorem{thm}{Theorem}[section]
\newtheorem{cor}[thm]{Corollary}
\newtheorem{lem}[thm]{Lemma}
\newtheorem{prop}[thm]{Proposition}

\theoremstyle{definition}

\newtheorem*{defns}{Definition}

\theoremstyle{remark}

\newtheorem*{rems}{Remark}

\newcommand{\vtsp}{\hspace{0.1em}}

\newlength{\phantomheight}

\DeclareMathOperator{\Dom}{Dom}
\DeclareMathOperator{\Ker}{Ker}

\DeclareMathOperator{\Ran}{Ran}

\DeclareMathOperator{\Span}{Sp}
\DeclareMathOperator{\sgn}{sgn}

\DeclareMathOperator{\Trace}{Tr}

\newcommand{\supp}{\mathrm{supp}}

\newcommand{\R}{\mathbb{R}}
\newcommand{\N}{\mathbb{N}}

\newcommand{\C}{\mathbb{C}}
\newcommand{\Z}{\mathbb{Z}}
\renewcommand{\epsilon}{\varepsilon}

\newcommand{\ipd}[2]{\langle{#1},{#2}\rangle}
\newcommand{\bigipd}[2]{\bigl\langle{#1},\vtsp{#2}\bigr\rangle}

\newcommand{\abs}[1]{\lvert{#1}\rvert}
\newcommand{\lrabs}[1]{\left\lvert{#1}\right\rvert}
\newcommand{\bigabs}[1]{\bigl\lvert{#1}\bigr\rvert}
\newcommand{\norm}[1]{\lVert{#1}\rVert}

\newcommand{\bignorm}[1]{\bigl\lVert{#1}\bigr\rVert}

\renewcommand{\d}{\,d}

\DeclareMathOperator{\spec}{spec}
\DeclareMathOperator{\Hom}{Hom}
\DeclareMathOperator{\Cl}{Cl}
\DeclareMathOperator{\trace}{tr}

\newcommand{\sphere}[1][]{\mathbb{S}_{#1}^2}  
\newcommand{\tsphere}{\mathbb{S}^3}  
\newcommand{\disc}[1][1]{\ifthenelse{\equal{#1}{1}}{\mathbb{D}}{\mathbb{D}_{#1}}}

\newcommand{\flux}[1]{\Phi(#1)}

\newcommand{\Dirac}[1]{\mathcal D_{#1}}
\newcommand{\sDirac}[2]{\mathcal D^{(#1)}_{#2}}
\newcommand{\Pauli}{\mathcal P}
\newcommand{\sD}{\mathcal D}  
\newcommand{\sA}{\mathcal A}  

\newcommand{\oP}{\mathbf{P}}
\newcommand{\oU}{\mathbf{U}}
\newcommand{\oQ}{\mathbf{Q}}
\newcommand{\oK}{\mathbf{K}}
\newcommand{\oL}{\mathbf{L}}
\newcommand{\oI}{\mathbf{I}}

\newcommand{\Zazm}{\mathbb{M}_\epsilon}
\newcommand{\Zbig}{\mathbb{M}_R'}
\newcommand{\PKazm}[1][\pm]{\Pi^{#1}_{\epsilon}}
\newcommand{\PKbig}{\Pi_{R}'}
\newcommand{\PKrem}{\Pi_{\epsilon,R}}

\newcommand{\ESf}[1][]{\mathcal B^{#1}_{\text{ES}}}

\newcommand{\cfn}[1]{\mathsf{N}_{#1}}
\newcommand{\cfnk}{\mathsf{N}_{B}^{(k)}}
\newcommand{\Mcfnk}{\mathsf{M}_{B}^{(k)}}
\newcommand{\cfD}[2][]{\ifthenelse{\equal{#2}{}}{\mathsf{n}_{#1}}{\mathsf{n}_{#1}(#2)}}
\newcommand{\dfD}[2][]{\mathsf{n}_{#1}(#2)}

\newcommand{\spin}{\mathcal H}  
\newcommand{\bspin}{\spin^2}  

\newcommand{\gaI}{J}
\newcommand{\gae}[2][\epsilon]{s^{#2}_{#1}}
\newcommand{\gaIe}[2][\epsilon]{J^{#2}_{#1}}

\newcommand{\ca}{a}
\newcommand{\Co}[1]{C_{#1}}  

\newcommand{\fipd}[3]{\langle{#2},{#3}\rangle_{#1}}

\newcommand{\fnorm}[2]{\lvert{#2}\rvert_{#1}}

\newcommand{\bigfnorm}[2]{\bigl\lvert{#2}\bigr\rvert_{#1}}

\newcommand{\stproj}[1]{z_{#1}}
\newcommand{\confwS}[1]{\Omega_{#1}}
\newcommand{\confwR}{\widetilde{\Omega}}
\newcommand{\dscmt}{\Delta}

\newcommand{\spbun}[1]{\ifthenelse{\equal{#1}{}}{\Psi}{\Psi^{(#1)}}}
\newcommand{\sLbun}[2]{\ifthenelse{\equal{#1}{}}{L_{#2}}{L^{(#1)}_{#2}}}
\newcommand{\detLbun}[1]{\ifthenelse{\equal{#1}{}}{L}{L^{(#1)}}}
\newcommand{\spc}[1][]{\widetilde{\nabla}_{#1}}
\newcommand{\LCc}[1][]{\nabla_{#1}}
\newcommand{\spf}[2]{\omega_{#1,#2}}

\newcommand{\volR}{\mathbf{v}_{\R^2}}
\newcommand{\volS}{\mathbf{v}_{\sphere}}
\newcommand{\formS}[1]{\Omega^{#1}(\sphere)}
\newcommand{\cont}[2][]{\ifthenelse{\equal{#1}{}}{C^{#2}}{C^{#2}_{#1}}}
\newcommand{\iof}[1]{\rho_{#1}}

\newcommand{\vb}[1]{e_{#1}}
\newcommand{\fb}[1]{\theta_{#1}}

\newcommand{\puS}[1]{\chi_{\delta,#1}}
\newcommand{\pxiS}[1]{\xi_{\delta,#1}}
\newcommand{\petaR}[1]{\eta_{\delta,#1}}

\title{Asymptotics for Erd\H{o}s-Solovej Zero Modes in Strong Fields}
\author{Daniel M.~Elton}

\begin{document}

\maketitle

\begin{abstract}
We consider the strong field asymptotics for the occurrence of zero modes of certain Weyl-Dirac operators on $\R^3$.
In particular we are interested in those operators $\Dirac{B}$ for which the associated magnetic field $B$
is given by pulling back a $2$-form $\beta$ from the sphere $\sphere$ to $\R^3$ 
using a combination of the Hopf fibration and inverse stereographic projection. If $\int_{\sphere}\beta\neq0$ we show that
\[
\sum_{0\le t\le T}\dim\Ker\Dirac{tB}
=\frac{T^2}{8\pi^2}\,\biggl\lvert\int_{\sphere}\beta\biggr\rvert\,\int_{\sphere}\abs{\beta}+o(T^2)
\]
as $T\to+\infty$. The result relies on Erd\H{o}s and Solovej's characterisation of the spectrum of $\Dirac{tB}$ in terms of a family of Dirac operators on $\sphere$, 
together with information about the strong field localisation of the Aharonov-Casher zero modes of the latter. 

\medskip

\end{abstract}

\section{Introduction}

Suppose $B$ is a (smooth) magnetic field on $\R^3$, 
viewed either as a divergence free vector field $B=(B_1,B_2,B_3)$ or as a closed $2$-form
\[
B=B_1\d x_2\wedge d x_3+B_2\d x_3\wedge d x_1+B_3\d x_1\wedge d x_2.
\]
Choose a corresponding magnetic potential (or $1$-form) $A=A_1\d x_1+A_2\d x_2+A_3\d x_3$ which generates $B$ in the sense that $B=dA$
(such potentials exist by Poincar\'e's Lemma).
A Weyl-Dirac operator operator can then be defined by
\begin{equation}
\label{WDR3:eq}
\Dirac{\R^3,B}
=\sum_{j=1}^3\sigma_j\vtsp(-i\nabla_j-A_j),
\end{equation}
where $\sigma_1$, $\sigma_2$ and $\sigma_3$ are the Pauli matrices and $\nabla=(\nabla_1,\nabla_2,\nabla_3)$ denotes the usual gradient operator on $\R^3$. The operator $\Dirac{\R^3,B}$ acts on $2$ component spinor-fields which, on $\R^3$, can be viewed simply as $\C^2$ valued functions. Standard arguments (see \cite[Theorem 4.3]{Th} for example) show that $\Dirac{\R^3,B}$ is essentially self-adjoint on $\C^\infty_0$. We also use $\Dirac{\R^3,B}$ to denote the corresponding closure which is an unbounded self-adjoint operator on $L^2(\R^3,\C^2)$.

We are interested in the question of when $0$ is an eigenvalue of $\Dirac{\R^3,B}$ or, equivalently, of determining when $\Dirac{\R^3,B}$ has a non-trivial kernel.

\begin{defns}
Any eigenfunction of $\Dirac{\R^3,B}$ corresponding to $0$ is called a \emph{zero mode}.
\end{defns}

\begin{rems}
The potential $A$ (and hence the operator $\Dirac{\R^3,B}$) is not uniquely determined by $B$.
However if $dA=B=dA'$ then $A-A'=d\phi$ for some $\phi\in\cont{\infty}(\R^3)$ (using Poincar\'e's Lemma). Multiplication by $e^{i\phi}$ then establishes a unitary equivalence between the operators $\Dirac{\R^3,B}$ defined using the potentials $A$ and $A'$. 
It follows that spectral properties of $\Dirac{\R^3,B}$, and in particular the existence of zero modes, depend only on $B$.
\end{rems}

Zero modes have been studied in a number of contexts in mathematical physics including the stability of matter (\cite{FLL}, \cite{LY}) and chiral gauge theories (\cite{AMN1}, \cite{AMN2}). 
Most early work concentrated on the construction of explicit examples, including the original example (\cite{LY}), examples with arbitrary multiplicity (\cite{AMN2}), compact support (\cite{E1}) and a certain rotational type of symmetry (\cite{ES}; further details below). Some subsequent work moved toward studying the set of all zero mode producing fields (or potentials) within a given class; in particular, this set is nowhere dense (\cite{BE1}, \cite{BE2}) and is  generically a co-dimension $1$ sub-manifold (\cite{E2}; slightly different classes of potentials were considered in these works). 

To further our understanding of which fields produce zero modes it is reasonable to consider the problem in various asymptotic regimes. 
We focus on the strong field regime (which, via a simple rescaling of the zero mode equation, is equivalent to the semi-classical regime).
For a fixed field $B$ define a counting function $\cfn{B}$ by
\[
\cfn{B}(T)=\sum_{0\le t\le T}\dim\Ker\Dirac{\R^3,tB}
\]
for any $T\in\R^+$. The behaviour of $\cfn{B}(T)$ as $T\to+\infty$ is more regular than that of $\dim\Ker\Dirac{\R^3,tB}$ and clearly gives information about the occurrence of zero modes for strong fields. 

In \cite{ET} an upper bound of the form $\cfn{B}(T)\le C\norm{A}_{L^3}^3T^3$ was obtained, valid for any $T\ge0$ and potential $A\in L^3$ (with $B=dA$). The purpose of the present work is to determine the precise leading order asymptotic behaviour of $\cfn{B}(T)$ as $T\to+\infty$ for a large class of symmetric magnetic fields first considered in \cite{ES}. Before defining this class we need to introduce some supporting ideas and notation.

Let $\formS{2}$ denote the set of $2$-forms on $\sphere$ and let $\volS\in\formS{2}$ denote the standard volume $2$-form. 
Any $\beta\in\formS{2}$ can then be written as $\beta=f\volS$ for a unique $f\in C^\infty(\sphere)$. The \emph{flux} of $\beta$ is defined to be
\[
\flux{\beta}=\frac1{2\pi}\int_{\sphere}\beta=\frac1{2\pi}\int_{\sphere}f\volS.
\]
We also define $\abs{\beta}$ to be the (not necessarily smooth) $2$-form given by $\abs{\beta}=\abs{f}\volS$. 

\begin{defns}
Let $h:\tsphere\to\sphere$ and $\pi:\tsphere\setminus\{(0,0,0,-1)\}\to\R^3$ denote the Hopf fibration and stereographic projection respectively. Set
\[
\ESf[\prime]=\left\{(\pi^{-1})^*h^*\beta:\beta\in\formS{2},\;\flux{\beta}\neq0\right\}
\]
(where ${}^*$ denotes pullback). 
Define $\ESf$ similarly except without the condition $\flux{\beta}\neq0$.
\end{defns}

Elements of $\ESf$ are closed $2$-forms on $\R^3$ and can thus be viewed as magnetic fields (note that, all $2$-forms on $\sphere$ are closed). 
Furthermore fields $B\in\ESf$ are smooth and satisfy bounds of the form $\abs{B(x)}=O(\abs{x}^{-4})$ as $\abs{x}\to\infty$, while it is always possible to find a smooth potential $A$ with $B=dA$ which satisfies bounds of the form $\abs{A(x)}=O(\abs{x}^{-3})$ as $\abs{x}\to\infty$. It follows that fields in $\ESf$ (and their associated potentials) fall into the classes considered in \cite{BE1}, \cite{BE2} and \cite{E2}.

\medskip

Our main result is the following.

\begin{thm}
\label{mainres1:thm}
Let $B\in\ESf[\prime]$ with $B=(\pi^{-1})^*h^*\beta$ for $\beta\in\formS{2}$. Then
\begin{equation}
\label{mainres:eq}
\cfn{B}(T)=\tfrac12\,\abs{\flux{\beta}}\,\flux{\abs{\beta}}\,T^2+o(T^2)
\quad\text{as $T\to+\infty$.}
\end{equation}
\end{thm}

The lower asymptotic bound in \eqref{mainres:eq}, together with the explicit form of $\cfn{B}(T)$ for the special case of the ``constant'' field $\beta=\volS$, were obtained in \cite{T}. It is also clear where the argument for the upper bound in \cite{ET} may gain an order in $T$, although it remains unclear whether the $O(T^3)$ upper bound might yet be sharp for some magnetic field $B$. 

Fields in $\ESf$ are invariant under the symmetry of $\R^3$ induced by the rotation of $\tsphere$ along the $\mathbb{S}^1$ fibres of the Hopf fibration. The main work in \cite{ES} is to show how this symmetry can be used to express the spectrum of $\Dirac{\R^3,tB}$ in terms of the spectra of a family of Dirac operators on $\sphere$ (see Section \ref{RtoS:sec} for further details). To calculate $\cfn{B}(T)$ we need to consider eigenvalues of the latter with modulus up to $1/4$. 
Aharonov-Casher zero modes (see Theorem \ref{ACS2:thm}) correspond to an eigenvalue of $0$ and 
contribute $\tfrac12\abs{\flux{\beta}}^2$ to the leading order coefficient on the right hand side of \eqref{mainres:eq}; when $\beta$ has a variable sign the remaining part of this coefficient comes from ``approximate zero modes'' which arise from the localising effects of strong fields (see Section \ref{azmlb:sec} for further details).

\medskip

This paper is organised as follows.
Some background on Dirac operators on $\sphere$ is outlined in Section \ref{DiracOpSphere:sec} while the key results we require from \cite{ES} 
are stated at the start of Section \ref{RtoS:sec}. The proof of Theorem \ref{mainres1:thm} is then reduced to determining the large $k$ 
asymptotics of a spectral quantity $\cfnk$ relating to a family of Dirac operators on $\sphere$; see \eqref{ZM2:eq} and Theorem \ref{typeIIest1:thm}.

The relatively straightforward lower bound in Theorem \ref{typeIIest1:thm} is covered in Section \ref{azmlb:sec}.
Necessary information about the asymptotic number of approximate zero modes for Dirac operators on $\sphere$ 
is given in Theorem \ref{basicappzmS2:thm} and justified in Section \ref{S2ctfn:sec} using equivalent results for the plane (from \cite{E3}).
Section \ref{azmlb:sec} concludes with further estimates relating to approximate zero modes;
some of the arguments rely on ideas from differential geometry and are deferred to Section \ref{S2spfldest:sec}.

The remaining sections are dedicated to the justification of the upper bound in Theorem \ref{typeIIest1:thm}.
In Section \ref{linear:sec} the quantity $\cfnk$ is expressed as the number of eigenvalues of a (non-self-adjoint) 
operator $\oL$ within a particular set; see Proposition \ref{evalLinfo:prop}. 
In turn this is estimated from the singular values of $\oL$ via Weyl's inequality; 
Section \ref{estsing:sec} is devoted to estimating the singular values while the argument is tied up in Section~\ref{mainarg:sec}.

\subsubsection*{Notation}

We use $\spec(T)$ to denote the set of eigenvalues of an operator $T$ with entries repeated according to geometric multiplicity.
The subset of positive eigenvalues is denoted by $\spec^+(T)$. 
General positive constants are denoted by $\Co{}$, with numerical subscripts used when we wish to keep track of specific constants in subsequent discussions.
The open disc in $\R^2$ with radius $r$ and centre $0$ is denoted $\disc[r]$, while $I_2$ denotes the $2\times 2$ identity matrix.

\section{Dirac operators on $\sphere$}
\label{DiracOpSphere:sec}

In order to discuss Dirac operators on $\sphere$ we firstly recall some
notions from Riemannian geometry as well as the idea of a spin$^c$ structure (spin$^c$ spinor bundles, Clifford multiplication and spin$^c$ connections). 
A fuller introduction can be found in \cite{F} (see also \cite{ES} for a discussion in a similar spirit to that presented here). 

Let $\fipd{\sphere}{\cdot}{\cdot}$ denote the standard Riemannian metric on (the tangent bundle of) $\sphere$, with corresponding norm $\fnorm{\sphere}{\cdot}$. The same symbols will be used for the induced metric on the exterior bundle $\wedge^*T^*\sphere$. For $n=0,1,2$ let $\formS{n}$ denote the set of $n$-forms (that is, sections of the $n$-form bundle $\wedge^nT^*\sphere$). 
Note that, $\int_{\sphere}\volS=4\pi$ while $\abs{\beta}=\fnorm{\sphere}{\beta}\volS$ for any $\beta\in\formS{2}$.

A \emph{spin$^c$ spinor bundle} $\spbun{}$ on $\sphere$ is a hermitian vector bundle over $\sphere$ with fibre $\C^2$ on which we can define Clifford multiplication. The latter is a unitary map $\sigma:T^*\sphere\to\Hom(\spbun{})$ which satisfies
\[
\sigma(\omega)\sigma(\rho)+\sigma(\rho)\sigma(\omega)=2\fipd{\sphere}{\omega}{\rho}I
\]
for all $1$-forms $\omega$ and $\rho$; here $\Hom(\spbun{})$ denotes the set of endomorphisms on $\spbun{}$ with inner product given by $\fipd{\Hom(\spbun{})}{A}{B}=\frac12\trace(A^*B)$, and $I\in\Hom(\spbun{})$ is the identity. (Clifford multiplication gives a unitary representation of the Clifford algebra $\Cl(T_x^*\sphere)$ on $\C^2$ which is isomorphic to the standard representation and varies smoothly with $x\in\sphere$.) 
Clifford multiplication extends naturally as a linear isomorphism $\sigma:\wedge^*T^*\sphere\to\Hom(\spbun{})$; in particular 
\begin{equation}
\label{sjsvcomm1:eq}
\sigma(\omega)\sigma(\volS)+\sigma(\volS)\sigma(\omega)=0
\end{equation}
for any $1$-form $\omega$, while $\sigma(\volS)^2=I$. The latter expression allows us to write $\spbun{}=\sLbun{}{+}\oplus\sLbun{}{-}$ where the line bundles $\sLbun{}{\pm}$ are defined by $\xi\in\sLbun{}{\pm}$ iff $\sigma(\volS)\xi=\pm\xi$. We use $\fipd{\spbun{}}{\cdot}{\cdot}$ and $\fnorm{\spbun{}}{\cdot}$ to denote the (fibrewise) inner-product and norm on $\spbun{}$, while $\Gamma(\spbun{})$ is the space of \emph{spinors} (sections of $\spbun{}$).

Associated to a spin$^c$ spinor bundle $\Psi$ is a line bundle which (for $\sphere$) is given as $\detLbun{}=\Psi\wedge\Psi$ (the determinant bundle of $\Psi$). This line bundle determines $\Psi$ up to isomorphism (note that, $H^2(\sphere;\Z)\cong\Z$ which has no $2$-torsion). On $\sphere$ there are infinitely many mutually non-isomorphic spin$^c$ spinor bundles which we denote as $\spbun{k}$ for $k\in\Z$, labelled so that the first Chern number of the associated line bundle satisfies $c_1(\detLbun{k})[\volS]=2k$. 

\medskip

Fix $k\in\Z$. A \emph{spin$^c$ connection} on $\spbun{k}$ is a connection $\spc$ which is compatible with hermitian structure on $\spbun{k}$ and the Clifford multiplication. For $\xi,\eta\in\Gamma(\spbun{k})$ and $X\in T\sphere$ the former compatibility means
\[
X\fipd{\spbun{k}}{\xi}{\eta}=\fipd{\spbun{k}}{\spc[X]\xi}{\eta}+\fipd{\spbun{k}}{\xi}{\spc[X]\eta},
\]
while the latter means $[\spc[X],\sigma(\omega)]=\sigma(\LCc[X]\omega)$
for all forms $\omega$; here $\LCc$ is the Levi-Civita connection on $\sphere$ (for the metric $\fipd{\sphere}{\cdot}{\cdot}$). 
As $\LCc[X]\volS=0$ we get
\begin{equation}
\label{spcsvcomm:eq}
[\spc[X],\sigma(\volS)]
=0.
\end{equation}

A spin$^c$ connection $\spc$ on $\spbun{k}$ is uniquely determined by a choice of (hermitian) connection on $\detLbun{k}$. It follows that the set of all spin$^c$ connections is an affine space modelled on $i\formS{1}$ (note that, $\detLbun{k}$ has structure group $U(1)$ with Lie algebra $i\R$). In particular, given $\spc$ any other spin$^c$ connection on $\spbun{k}$ can be written as $\spc-i\alpha$ for some $\alpha\in\formS{1}$.

The curvature of the connection $\spc$ can be viewed as the $\Hom(\spbun{k})$ valued $2$-form given by
\[
\widetilde{R}(X,Y)\xi
=\spc[X]\spc[Y]\xi-\spc[Y]\spc[X]\xi-\spc[{[X,Y]}]\xi
\]
for all $X,Y\in T\sphere$ and $\xi\in\spbun{k}$. The \emph{magnetic $2$-form} of $\spc$ is then defined to be $\beta=\frac{i}2\Trace(\widetilde{R})\in\formS{2}$. The first Chern class of $\detLbun{k}$ is the cohomology class of $\frac1\pi\beta$ so
\[
\flux{\beta}
=\frac1{2\pi}\int_{\sphere}\beta
=\tfrac12\,c_1(\detLbun{k})[\volS]=k;
\]
that is, the total flux of any magnetic $2$-form on $\spbun{k}$ must be equal to $k$. This flux condition is also sufficient for a $2$-form to be the magnetic $2$-form of a spin$^c$ connection on $\spbun{k}$. More precisely if $\beta'\in\formS{2}$ with $\flux{\beta'}=k$ then $\beta'=\beta+d\alpha$ for some $\alpha\in\formS{1}$ (this follows from the Hodge decomposition theorem and 
the fact that the harmonic $2$-forms on $\sphere$ are simply the constant multiples of $\volS$). 
A straightforward calculation then shows $\beta'$ is the magnetic $2$-form associated to the spin$^c$ connection $\spc'=\spc-i\alpha$. The choice of $\alpha$ is only unique up to the addition of a closed $1$-form. 

\medskip

The \emph{Dirac operator} corresponding to a 
given a spin$^c$ connection $\spc$ on $\spbun{k}$ is defined as $\Dirac{}=-i\Trace\sigma\spc$. 
If $\{\vb{1},\vb{2}\}$ is a local orthonormal frame (of vector fields) 
with corresponding dual frame $\{\fb{1},\fb{2}\}$ (of $1$-forms) we can equivalently write
\[
\Dirac{}
=-i\sigma(\fb{1})\spc[\vb{1}]-i\sigma(\fb{2})\spc[\vb{2}].
\]
The operator $\Dirac{}$ maps $\Gamma(\spbun{k})\to\Gamma(\spbun{k})$. Taking closures $\Dirac{}$ becomes a(n unbounded) self-adjoint operator on the $L^2$ sections of $\spbun{k}$; we denote the latter by $\spin$. Since $\Dirac{}$ is a first order elliptic differential operator on a compact manifold it has a compact resolvent and discrete spectrum. Furthermore \eqref{sjsvcomm1:eq} and \eqref{spcsvcomm:eq} give
\begin{equation}
\label{Diracsvolacomm:eq}
\Dirac{}(\sigma(\volS)\,\cdot\,)=-\sigma(\volS)\Dirac{},
\end{equation}
so the spectrum of $\Dirac{}$ is symmetric about $0$. Combined with the Aharonov-Casher theorem (\cite{AC}; see \cite{ES} for the $\sphere$ version) we then have the following.

\begin{thm}
\label{ACS2:thm}
For any Dirac operator $\Dirac{}$ on $\spbun{k}$ we have $\dim\Ker\Dirac{}=\abs{k}$, while the spectrum of $\Dirac{}$ is symmetric about $0$.
\end{thm}

\begin{rems}
For the decomposition $\spbun{k}=\sLbun{k}{+}\oplus\sLbun{k}{-}$ (induced by $\sigma(\volS)$) \eqref{Diracsvolacomm:eq} leads to
\[
\Dirac{}=\begin{pmatrix}0&\Dirac{-}\\\Dirac{+}&0\end{pmatrix}
\]
with $\Dirac{\pm}:\Gamma(\sLbun{k}{\pm})\to\Gamma(\sLbun{k}{\mp})$. 
The Aharonov-Casher theorem can then be viewed as a combination of the Atiyah-Singer index theorem and a vanishing theorem for $\Dirac{}$; the former gives
\[
\dim\Ker\Dirac{+}-\dim\Ker\Dirac{-}=\tfrac12c_1(\detLbun{k})[\volS]=k,
\]
while the latter forces either $\Ker\Dirac{+}$ or $\Ker\Dirac{-}$ to be trivial.
\end{rems}

A straightforward calculation shows that the Dirac operator associated to the spin$^c$ connection $\spc'=\spc-i\alpha$ is $\Dirac{}'=\Dirac{}-\sigma(\alpha)$. Dirac operators also satisfy a simple gauge transformation rule; if $\psi\in C^\infty(\sphere)=\formS{0}$ then
\[
e^{i\psi}\Dirac{}(e^{-i\psi}\cdot)=\Dirac{}-\sigma(d\psi),
\]
the Dirac operator corresponding to the spin$^c$ connection $\spc-id\psi$. In particular the Dirac operators corresponding to the spin$^c$ connections $\spc$ and $\spc-id\psi$ are unitarily equivalent and hence have the same spectrum. It follows that the spectrum of a Dirac operator on $\sphere$ is determined entirely by the magnetic $2$-form of the corresponding spin$^c$ connection (note that $H^1(\sphere)=0$ so $d\formS{0}$ is precisely the set of closed $1$-forms). 

Let $\spc^{(k)}$ denote a spin$^c$ connection on $\spbun{k}$ corresponding to the ``constant'' magnetic $2$-form $\frac{k}2\volS$ and let $\sDirac{k}{}$ denote the corresponding Dirac operator.
If $\beta\in\formS{2}$ is any other $2$-form with $\flux{\beta}=\flux{\frac{k}2\volS}=k$ we can find $\alpha\in\formS{1}$ with $\beta=\frac{k}2\volS+d\alpha$ (as above). The spin$^c$ connection $\spc^{(k)}-i\alpha$ then has magnetic $2$-form $\beta$ and corresponding Dirac operator 
\begin{equation}
\label{pertformDirac:eq1}
\sDirac{k}{\alpha}=\sDirac{k}{}-\sigma(\alpha). 
\end{equation}
This operator is uniquely determined by $\beta$ up to gauge (and hence unitary) equivalence. 
We can view $\alpha$ as generating the ``non-constant'' part of $\beta$.

\bigskip

The situation for Dirac operators on $\tsphere$ is rather simpler. All Spin$^c$ bundles on $\tsphere$ are isomorphic to the trivial bundle $\tsphere\times\C^2$, while any closed $2$-form $b\in\Omega^2(\tsphere)$ gives rise to a self-adjoint Dirac operator $\Dirac{\tsphere,b}$, which is unique up to unitary equivalence; see \cite{ES} for further details.

\section{Reduction to $\sphere$}

\label{RtoS:sec}

Let $\beta\in\formS{2}$ with $\flux{\beta}=1$. From the above discussion we can write $\beta=\frac12\volS+d\alpha$ for some $\alpha\in\formS{1}$. Also set $b=h^*\beta$, the closed $2$-form on $\tsphere$ obtained by pulling back $\beta$ using the Hopf fibration $h:\tsphere\to\sphere$. For $t\in\R$ the magnetic field $tb$ is invariant under rotations of $\tsphere$ along the level sets of $h$. 
This symmetry is inherited by the Dirac operator $\Dirac{\tsphere,tb}$, which allows the spectrum of $\Dirac{\tsphere,tb}$ to be expressed in terms of the spectra of a family of Dirac operators on $\sphere$. 
The following is a restatement of \cite[Theorem 8.1]{ES} 
(note that the metric $\tfrac14\fipd{\sphere}{\cdot}{\cdot}$ is used in \cite{ES} so eigenvalues of Dirac operators on $\sphere$ must include an extra factor of $2$ here). 

\begin{thm}
\label{ESres1:thm}
For any $t\in\R$ the spectrum of $\Dirac{\tsphere,tb}$ is
\[
\bigcup_{k\in\Z}\;\Sigma_k\cup
\bigl\{-\tfrac12+\sqrt{4\lambda^2+(k-t)^2},\,-\tfrac12-\sqrt{4\lambda^2+(k-t)^2}\,:\,\lambda\in\spec^+\bigl(\sDirac{k}{t\alpha}\bigr)\bigr\}
\]
where $\Sigma_k$ contains the number $-\tfrac12-\sgn(k)\vtsp(k-t)$ counted with multiplicity $\abs{k}$ (so $\Sigma_0=\emptyset$). 
The multiplicity of an eigenvalue of $\Dirac{\tsphere,tb}$ is equal to the number of times it appears in the above list when the elements of $\Sigma_k$ and $\spec^+(\sDirac{k}{t\alpha})$ are counted with their relevant multiplicities.
\end{thm}

Set $B=(\pi^{-1})^*b=(\pi^{-1})^*h^*\beta\in\ESf[\prime]$. From \cite[Theorem 8.7]{ES} we have the following link between the Dirac operators $\Dirac{\tsphere,tb}$ and $\Dirac{\R^3,tB}$.

\begin{thm}
\label{ESres2:thm}
For any $t\in\R$ we have $\dim\Ker\Dirac{\R^3,tB}=\dim\Ker\Dirac{\tsphere,tb}$.
\end{thm}

Consider the disjoint partition of $\R$ given by the intervals
\[
\widetilde{\tau}_k=\begin{cases}
(k-\tfrac12,k+\tfrac12]&\text{if $k>0$,}\\
[-\tfrac12,\tfrac12]&\text{if $k=0$,}\\
[k-\tfrac12,k+\tfrac12)&\text{if $k<0$,}     
\end{cases}
\]
for $k\in\Z$. Also let $\tau_k=(k-1/2,k+1/2)$ and $\overline{\tau}_k=[k-1/2,k+1/2]$ 
denote the interior and closure of $\widetilde{\tau}_k$ respectively.
To identify the contribution to $\cfn{B}$ 
coming from $t\in\tau_k$ and $t\in\widetilde{\tau}_k$ set
\[
\Mcfnk=\sum_{t\in\tau_k}\dim\Ker\Dirac{\R^3,tB}
\quad\text{and}\quad
\cfnk=\sum_{t\in\widetilde{\tau}_k}\dim\Ker\Dirac{\R^3,tB}.
\]

From Theorems \ref{ESres1:thm} and \ref{ESres2:thm} it is clear that $\Ker\Dirac{\R^3,tB}$ is non-trivial precisely when there exists $k\in\Z$ such that either $0\in\Sigma_k$ or $4\lambda^2+(k-t)^2=1/4$ for some $\lambda\in\spec^+(\sDirac{k}{t\alpha})$, with corresponding agreement of multiplicities. 
In the latter case we have $\lambda>0$ which forces $(k-t)^2<1/4$ or $t\in\tau_k$. It follows that
\[
\Mcfnk
=\#\bigl\{(t,\lambda):\text{$\lambda\in\spec^+\bigl(\sDirac{k}{t\alpha}\bigr)$ and $4\lambda^2+(k-t)^2=\tfrac14$}\bigr\}.
\]
We also know that $0$ is contained in the spectrum of $\sDirac{k}{t\alpha}$ with multiplicity $\abs{k}$ for any $t\in\R$ (see Theorem \ref{ACS2:thm}), while $0^2+(k-t)^2=1/4$ has two solutions ($t=k\pm1/2$). Furthermore the spectrum of $\sDirac{k}{t\alpha}$ is symmetric about $0$. Combining these observations we get
\[
\#\bigl\{(t,\lambda):\text{$\lambda\in\spec\bigl(\sDirac{k}{t\alpha}\bigr)$ and $4\lambda^2+(k-t)^2=\tfrac14$}\bigr\}
=2\Mcfnk+2\abs{k}.
\]
On the other hand $0\in\Sigma_k$, with multiplicity $\abs{k}$, iff $t\in\widetilde{\tau}_k\setminus\tau_k$. 
It follows that $\cfnk-\Mcfnk=\abs{k}$ and so
\begin{equation}
\label{ZM2:eq}
\cfnk=\frac12\#\bigl\{(t,\lambda):\text{$\lambda\in\spec\bigl(\sDirac{k}{t\alpha}\bigr)$ and $4\lambda^2+(k-t)^2=\tfrac14$}\bigr\}.
\end{equation}
Clearly in calculating the right hand side of \eqref{ZM2:eq} we need only consider $t\in\overline{\tau}_k$ and eigenvalues of $\sDirac{k}{t\alpha}$ in $[-1/4,1/4]$. In addition to the eigenvalue $0$ with multiplicity $\abs{k}$ (the Aharonov-Casher zero modes) there may be small non-zero eigenvalues (the approximate zero modes). The total number of these eigenvalues can be determined asymptotically in $\abs{k}$ (see Theorem \ref{basicappzmS2:thm}) which ultimately leads to the following.

\begin{thm}
\label{typeIIest1:thm}
We have $\cfnk=\flux{\abs{\beta}}\,\abs{k}+o(\abs{k})$ as $\abs{k}\to\infty$.
\end{thm}

The lower bound for $\cfnk$ contained in Theorem \ref{typeIIest1:thm} was given in \cite{T} and is included here for completeness (see Section \ref{azmlb:sec}). The justification of the upper bound for $\cfnk$ appears in Section \ref{mainarg:sec}.

\smallskip

Given Theorem \ref{typeIIest1:thm} the proof of our main result is now straightforward.

\begin{proof}[Proof of Theorem \ref{mainres1:thm}]
We can extend the definition of $\cfn{B}(T)$ to cover $T<0$ by summing over $T\le t\le0$ in this case. Together with the scaling properties of \eqref{mainres:eq} it thus suffices to restrict to the case $\flux{\beta}=1$ and prove
\begin{equation}
\label{mainres:v1:eq}
\cfn{B}(T)=\tfrac12\,\flux{\abs{\beta}}\,T^2+o(T^2)
\quad\text{as $T\to\pm\infty$.}
\end{equation}
Now let $T>0$ and pick $k_T\in\Z$ with $T\in\widetilde{\tau}_{k_T}$. Then 
$\bigcup_{k=1}^{k_T-1}\widetilde{\tau}_k\subset[0,T]\subset\bigcup_{k=0}^{k_T}\widetilde{\tau}_k$ so
\[
\sum_{k=1}^{k_T-1}\cfnk\le\cfn{B}(T)\le\sum_{k=0}^{k_T}\cfnk.
\]
Using Theorem \ref{typeIIest1:thm} and the fact that $\abs{k_T-T}\le1/2$ we get
\[
\sum_{k=0}^{k_T}\cfnk
=\sum_{k=0}^{k_T}\vtsp[\flux{\abs{\beta}}\,k+o(k)]
=\tfrac12\,\flux{\abs{\beta}}\,k_T^2+o(k_T^2)
=\tfrac12\,\flux{\abs{\beta}}\,T^2+o(T^2)
\]
as $T\to+\infty$. Since the removal of the first and last terms from the sum will not change this asymptotic \eqref{mainres:v1:eq} for $T>0$ now follows. A similar argument clearly deals with the case $T<0$.
\end{proof}

\section{The lower bound}

\label{azmlb:sec}

Throughout the next four sections we consider a 
fixed $\beta\in\formS{2}$ with $\flux{\beta}=1$ and write $\beta=\frac12\volS+d\alpha$ for some $\alpha\in\formS{1}$. 
For each $k\in\Z$ and $\epsilon,R>0$ set
\[
\cfD{\epsilon}
=\cfD[k,\alpha]{\epsilon}
=\#\bigl\{\lambda\in\spec\bigl(\sDirac{k}{k\alpha}\bigr):\abs{\lambda}\le\epsilon\bigr\}
\]
(counting according to multiplicity) and
\[
\dfD{\epsilon,R}
=\dfD[k,\alpha]{\epsilon,R}
=\begin{cases}
\cfD{R}-\cfD{\epsilon}&\text{if $R\ge\epsilon$,}\\
0&\text{if $R<\epsilon$.}
\end{cases}
\]
Since $\sDirac{k}{k\alpha}$ has $\abs{k}$ zero modes (recall Theorem \ref{ACS2:thm}) we have $\cfD{\epsilon}\ge\abs{k}=\abs{\flux{k\beta}}$;
a strict inequality (for suitable $\epsilon$) reflects the presence of approximate zero modes. 
In general there will be $O(\abs{k})$ approximate zero modes whenever $\beta$ has variable sign; 
more precisely we have the following.

\begin{thm}
\label{basicappzmS2:thm}
Suppose $\epsilon_k=Ce^{-c\abs{k}^\rho}$ for some $C,c>0$ and $0<\rho<1$, while $R_k=o(\abs{k}^{1/2})$ as $\abs{k}\to\infty$. 
Then 
\[
\liminf_{\abs{k}\to\infty}\frac1{\abs{k}}\,\cfD[k,\alpha]{\epsilon_k}\ge\flux{\abs{\beta}}
\quad\text{and}\quad
\limsup_{\abs{k}\to\infty}\frac1{\abs{k}}\,\cfD[k,\alpha]{R_k}\le\flux{\abs{\beta}}.
\]
Consequently $\cfD[k,\alpha]{\epsilon_k}=\flux{\abs{\beta}}\vtsp\abs{k}+o(\abs{k})$ and $\dfD[k,\alpha]{\epsilon_k,R_k}=o(\abs{k})$ as $\abs{k}\to\infty$.
\end{thm}

The proof of this result is given in Section \ref{S2ctfn:sec} where it is reduced to a similar result for the Pauli operator on a disc in $\R^2$.

\medskip

From \eqref{pertformDirac:eq1} we get
\begin{equation}
\label{pertformDirac:eq2}
\sDirac{k}{t\alpha}=\sDirac{k}{}-t\sigma(\alpha). 
\end{equation}
It follows that $t\mapsto\sDirac{k}{t\alpha}$ defines a self-adjoint holomorphic family of operators. 
Using standard perturbation theory (see \cite{K}) we can then choose real-analytic functions $\mu_n$ for $n\in\Z$ so that
the full set of eigenvalues of $\sDirac{k}{t\alpha}$ (including multiplicities) is $\{\mu_n(t):n\in\Z\}$ for any $t\in\R$.
We can now rewrite \eqref{ZM2:eq} as
\begin{equation}
\label{ZM3:eq}
\cfnk=\frac12\,\#\bigl\{(n,t)\in\Z\times\R\,:\,4\mu^2_n(t)+(t-k)^2=\tfrac14\bigr\}.
\end{equation}

\begin{proof}[Proof of lower bound in Theorem \ref{typeIIest1:thm}]
Fix $\epsilon\in(0,1/4)$ and suppose 
$\abs{\mu_n(k)}\le\epsilon$ for some $n\in\Z$. Then $4\mu^2_n(k)+(k-k)^2\le4\epsilon^2<1/4$. 
However $\mu_n$ is continuous and $4\mu^2_n(k\pm1/2)+(k\pm1/2-k)^2\ge1/4$, so  
there are at least two values of $t$ with $4\mu^2_n(t)+(t-k)^2=1/4$. 
From \eqref{ZM3:eq} it follows that
\begin{equation}
\label{keyestlb:ieq}
\cfnk\ge\#\bigl\{n\in\Z:\abs{\mu_n(k)}\le\epsilon\bigr\}
=\cfD[k,\alpha]{\epsilon}.
\end{equation}
The lower bound in Theorem \ref{typeIIest1:thm} now follows from Theorem \ref{basicappzmS2:thm}.
\end{proof}

The complication with obtaining the upper bound in Theorem \ref{typeIIest1:thm} is that, for each $n\in\Z$ with $\mu_n(k)<1/4$, we need upper bounds on the number of values of $t$ with $4\mu^2_n(t)+(t-k)^2=1/4$; in general there is no reason why this can't be more than two. 
We need some information about how rapidly $\mu_n(t)$ can change with respect to $t$. 

\begin{prop}
\label{spinalign1:prop}
For $j=1,2$ suppose $\lambda_j$ is an eigenvalue of $\sDirac{k}{t\alpha}$ with normalised eigenfunction $\xi_j$. 
Then $\abs{\ipd{\xi_1}{\sigma(\alpha')\xi_2}}\le\pi(\abs{\lambda_1}+\abs{\lambda_2})\,\norm{\alpha'}_{L^\infty}$ for any $\alpha'\in\Omega^1(\sphere)$.
\end{prop}

The proof of this result is given in Section \ref{S2spfldest:sec}.

\begin{rems}
If $\xi\in\Ker\sDirac{k}{t\alpha}$ Proposition \ref{spinalign1:prop} gives $\ipd{\xi}{\sigma(\alpha')\xi}=0$ for any $\alpha'\in\Omega^1(\sphere)$, 
which forces the value of $\xi$ to lie in either $\sLbun{k}{+}$ or $\sLbun{k}{-}$ at each point of $\sphere$. 
This result can be viewed as a local version of the vanishing theorem underlying the Aharonov-Casher theorem.
\end{rems}

\begin{cor}
\label{muntrelmunk:cor}
Set $\ca=2\pi\norm{\alpha}_{L^\infty}$. For any $n\in\Z$ and $t\in\R$ we have
\[
e^{-\ca\abs{t-k}}\abs{\mu_n(k)}
\le\abs{\mu_n(t)}
\le e^{\ca\abs{t-k}}\abs{\mu_n(k)}.
\]
\end{cor}

\begin{proof}
Fix $n$. Since $\sDirac{k}{t\alpha}$ is a self-adjoint holomorphic family we can choose a normalised eigenfunction $\xi(t)$ for $\mu_n(t)$ which is real-analytic in $t$ (see \cite{K}). Applying standard first order perturbation theory to \eqref{pertformDirac:eq2} then gives
\[
\frac{d}{dt}\mu_n(t)=-\ipd{\xi(t)}{\sigma(\alpha)\xi(t)}.
\]
Thus $\abs{d\mu_n/dt}\le a\abs{\mu_n}$ by Proposition \ref{spinalign1:prop}.
Integration completes the result.
\end{proof}

Let $\epsilon>0$ and suppose $\abs{\mu_n(k)}\le\epsilon$ (ultimately we will use $\epsilon$ to control the size of the approximate zero modes of $\sDirac{k}{k\alpha}$). For sufficiently small $\epsilon$ Corollary \ref{muntrelmunk:cor} provides enough control over the behaviour of $\mu_n(t)$ when $t\in\overline{\tau}_k$ to ensure that there are precisely two values of $t$ with $4\mu^2_n(t)+(t-k)^2=1/4$. Therefore the issue of extra values of $t$ can only arise when $\epsilon<\abs{\mu_n(k)}\le1/4$. For reasonable choices of $\epsilon$ Theorem \ref{basicappzmS2:thm} shows there are at most $o(\abs{k})$ such eigenvalues; we need to show that these eigenvalues lead to at most $o(\abs{k})$ extra values of $t$.

\section{Linearisation}

\label{linear:sec}

Our aim (Proposition \ref{evalLinfo:prop}) is to re-express the quantity $\cfnk$ as the number of eigenvalues of some 
(compact non-self-adjoint) operator $\oL$ within a prescribed set. 
In essence this is achieved by using \eqref{pertformDirac:eq2} and \eqref{ZM3:eq} to view $\cfnk$ 
as the number of real eigenvalues of a quadratic spectral pencil 
and then linearising this pencil by moving to a suitably chosen $2\times 2$ system.

\medskip

Introduce a shifted parameter $s=t-k+1$. 
Then $t\in\overline{\tau}_k$ iff $s\in\gaI$ where $\gaI=[1/2,3/2]$.
Also set $\sD=\sDirac{k}{(k-1)\alpha}$ and $\sA=\sigma(\alpha)$ so \eqref{pertformDirac:eq2} becomes
\[
\sDirac{k}{t\alpha}=\sD-s\sA.
\]
  
Let $\oI=I\otimes I_2$ denote the identity on $\bspin=\spin\otimes\C^2$.
Introduce further operators $\oP$ and $\oQ=\oQ_0+\oQ_1$ on $\bspin$ where
\begin{gather*}
\oP=2\sD\otimes\sigma_3+I\otimes\sigma_1-\tfrac12\oI=\begin{pmatrix}2\sD-\tfrac12I&I\\I&-2\sD-\tfrac12I\end{pmatrix},\\
\oQ_0=I\otimes\sigma_1=\begin{pmatrix}0&I\\I&0\end{pmatrix}
\quad\text{and}\quad
\oQ_1=2\sA\otimes\sigma_3=\begin{pmatrix}2\sA&0\\0&-2\sA\end{pmatrix}.
\end{gather*}
In particular
\[
\oP-s\oQ=\begin{pmatrix}2(\sD-s\sA)&(1-s)I\\(1-s)I&-2(\sD-s\sA)\end{pmatrix}-\tfrac12\oI.
\]
The operators $\oP$ and $\oQ$ are self-adjoint with $\oP$ unbounded and $\oQ$ bounded. 
In particular $\Dom\oP=(\Dom\sD)^2$ while $\oP-s\oQ$ has a compact resolvent for any $s\in\R$ (as $\sD-s\sA$ does).
Also
\begin{equation}
\label{PsQ2:eq}
(\oP-s\oQ+\tfrac12\oI)^2=\bigl[4(\sD-s\sA)^2+(s-1)^2I\bigr]\otimes I_2.
\end{equation}
Taking $s=0$ we get $(\oP+\oI/2)^2\ge\oI$ so $\abs{\oP}\ge\oI/2$ where $\oP=\abs{\oP}\oU$ is the polar decomposition of $\oP$.
It follows that $\abs{\oP}^{-1/2}$ is an injective compact operator with $\norm{\abs{\oP}^{-1/2}}\le\sqrt2$.
Define a further compact operator by 
\[
\oL=\oU\abs{\oP}^{-1/2}\oQ\abs{\oP}^{-1/2}.
\]

Let $\Co{1}=4e^{\ca/2}$. 
For $0\le\epsilon\le1/\Co{1}$ set $\gae{\pm}=1\pm(1-\Co{1}\epsilon)/2$ so $\gaI=[\gae[0]{-},\gae[0]{+}]$. 
Also set $\gaIe{+}=[\gae{+},\gae[0]{+}]$ and $\gaIe{-}=[\gae[0]{-},\gae{-}]$.

\begin{prop}
\label{evalLinfo:prop}
We have
\begin{equation}
\label{ZM5:eq}
\cfnk=\frac12\,\#\bigl\{\lambda\in\spec(\oL)\,:\,\lambda^{-1}\in\gaI\bigr\}.
\end{equation}
Furthermore if $0<\epsilon\le1/\Co{1}$ then 
\begin{equation}
\label{azmLest:eq}
\#\bigl\{\lambda\in\spec(\oL)\,:\,\lambda^{-1}\in\gaIe{\pm}\bigr\}
\ge\cfD{\epsilon}.
\end{equation}
\end{prop}

Approximate zero modes correspond to the eigenvalues of $\oL$ with reciprocals in $\gaIe{+}$ and $\gaIe{-}$; \eqref{azmLest:eq} is the corresponding restatement of \eqref{keyestlb:ieq}.

\begin{proof}
From \eqref{ZM3:eq} and \eqref{PsQ2:eq} we get
\begin{align*}
\cfnk&=\frac12\sum_{s\in J}\#\bigl\{n\in\Z\vtsp:\vtsp4\mu^2_n(s+k-1)+(s-1)^2=\tfrac14\bigr\}\\
&=\frac14\sum_{s\in J}\dim\Ker\bigl[(\oP-s\oQ+\tfrac12\oI)^2-\tfrac14\oI\bigr].
\end{align*}
Now $(I\otimes\sigma_2)(\oP-s\oQ)(I\otimes\sigma_2)=-(\oP-s\oQ)-\oI$ (note that $\sigma_2^2=I_2$ while $\sigma_2\sigma_j\sigma_2=-\sigma_j$ for $j=1,3$).
It follows that
\begin{align*}
\dim\Ker\bigl[(\oP-s\oQ+\tfrac12\oI)^2-\tfrac14\oI\bigr]
&=\dim\Ker(\oP-s\oQ)+\dim\Ker(\oP-s\oQ+\oI)\\
&=2\dim\Ker(\oP-s\oQ).
\end{align*}
However $\oI-s\oL=\oU\abs{\oP}^{-1/2}(\oP-s\oQ)\abs{\oP}^{-1/2}$ 
so $\dim\Ker(\oI-s\oL)=\dim\Ker(\oP-s\oQ)$ for any $s$ (recall that $\abs{\oP}^{-1/2}$ is injective).
Combining the above gives \eqref{ZM5:eq}.

\smallskip

Now $\abs{\gae{\pm}-1}=(1-\Co{1}\epsilon)/2\le1/2$. If $\abs{\mu_n(k)}\le\epsilon\le1/\Co{1}$ for some $n\in\Z$ then
\[
\abs{\mu_n(k+\gae{\pm}-1)}\le e^{\ca/2}\vtsp\abs{\mu_n(k)}
\le\tfrac14\Co{1}\epsilon
\]
using Corollary \ref{muntrelmunk:cor}. It follows that
\begin{align*}
4\mu^2_n(k+\gae{\pm}-1)+(\gae{\pm}-1)^2
\le\tfrac14(\Co{1}\epsilon)^2+\tfrac14(1-\Co{1}\epsilon)^2
\le\tfrac14.
\end{align*}
However $4\mu^2_n(k+\gae[0]{\pm}-1)+(\gae[0]{\pm}-1)^2\ge1/4$ while $\mu_n$ is continuous. 
Thus there is at least one $s\in\gaIe{\pm}$ with $4\mu^2_n(k+s-1)+(s-1)^2=1/4$.
Since $\cfD{\epsilon}=\#\{n\in\Z:\abs{\mu_n(k)}\le\epsilon\}$ estimate \eqref{azmLest:eq} now follows.
\end{proof}

\section{Estimates for singular values}

\label{estsing:sec}

Define compact self-adjoint operators by 
\[
\oK_j=\abs{\oP}^{-1/2}\oQ_j\abs{\oP}^{-1/2},\quad j=0,1
\]
and $\oK=\oK_0+\oK_1$. Then $\oL=\oU\oK$ so $\oL^*\oL=\oK^2$;
in particular, the singular values of $\oL$ are simply the moduli of the eigenvalues of $\oK$.
In order to study the latter we treat $\oK_1$ as a perturbation of $\oK_0$; in turn, the spectrum of $\oK_0$ can be determined from that of $\sD$.

\medskip

For any $d\in\R$ let $X_d$ denote the symmetric $2\times2$ matrix
\[
X_d=\begin{pmatrix}2d-\tfrac12&1\\1&-2d-\tfrac12\end{pmatrix}.
\]
The eigenvalues of $X_d$ are $-1/2+\dscmt$ and $-1/2-\dscmt$ where $\dscmt=\sqrt{4d^2+1}\ge1$. Thus
\begin{equation}
\label{normXdiest:eq}
\norm{\abs{X_d}^{-1/2}}=\bigl(\dscmt-\tfrac12\bigr)^{-1/2}\le\min\{\sqrt2,\,\abs{d}^{-1/2}\}.
\end{equation}
Define a quadratic polynomial by
\[
p_d(\lambda)=\lambda^2+\dscmt^{-1}\,\lambda+\tfrac14-\dscmt^2.
\]
Then $p_d(0)\le-3/4$ so $p_d$ has one root of each sign; let $\kappa^\pm(d)$ denote the reciprocal of the root with sign $\pm1$.
Note that $\kappa^+(0)=2$ and $\kappa^-(0)=-2/3$.

\begin{lem}
\label{behkappa:lem}
The eigenvalues of the $2\times2$ matrix $\abs{X_d}^{-1/2}\sigma_1\abs{X_d}^{-1/2}$ are $\kappa^+(d)$ and $\kappa^-(d)$.
Furthermore $\pm\kappa^\pm(d)\le\min\{2,\abs{d}^{-1}\}$ and $\abs{\kappa^\pm(d)-\kappa^\pm(0)}\le16d^2$.
\end{lem}

Let $x^\pm_d\in\C^2$ denote a normalised eigenvector of $\abs{X_d}^{-1/2}\sigma_1\abs{X_d}^{-1/2}$ corresponding to $\kappa^\pm(d)$.

\begin{proof}
We have $\det(\abs{X_d}\sigma_1)=-\abs{\det X_d}=\tfrac14-\dscmt^2$ while $2\dscmt\abs{X_d}+X_d=(2\dscmt^2-\tfrac12)I_2$ so
$\Trace(2\dscmt\abs{X_d}\sigma_1)=-\Trace(X_d\sigma_1)=-2$.
Thus $p_d$ is the characteristic polynomial of $\abs{X_d}\sigma_1$ and hence $\abs{X_d}^{1/2}\sigma_1\abs{X_d}^{1/2}$. 
The first part of the result follows as $\sigma_1^{-1}=\sigma_1$, while the second part can then be obtained from \eqref{normXdiest:eq} and the fact that $\norm{\sigma_1}=1$.

\smallskip

Let $\chi^\pm_d=1/\kappa^\pm(d)$ denote the roots of $p_d$; in particular $\abs{\chi^\pm_d}\ge1/2$. 
Now $p_d(\lambda)$ is decreasing in $d^2$ for fixed $\lambda>0$ 
so $\chi^+_d\ge\chi^+_0=1/2$.
Also 
\[
p_d\bigl(\dscmt^2-\tfrac12\bigr)
\ge\dscmt-\tfrac12(1+\dscmt^{-1})\ge0 
\]
(recall that $\dscmt\ge1$) so $\chi^+_d\le\dscmt^2-1/2$.
Thus $0\le\chi^+_d-\chi^+_0\le\dscmt^2-1=4d^2$.
On the other hand $\chi^+_d+\chi^-_d=-\dscmt^{-1}$ for any $d$ so
\[
(\chi^+_d-\chi^+_0)+(\chi^-_d-\chi^-_0)
=1-\dscmt^{-1}\in[0,\vtsp2d^2].
\]
It follows that $\abs{\chi^-_d-\chi^-_0}\le4d^2$. Combined we then get
\[
\abs{\kappa^\pm(d)-\kappa^\pm(0)}
=\frac{\abs{\chi^\pm_d-\chi^\pm_0}}{\abs{\chi^\pm_d}\,\abs{\chi^\pm_0}}
\le\frac{4d^2}{(\tfrac12)(\tfrac12)}=16d^2,
\]
completing the result.
\end{proof}

Set $\nu_n=\mu_n(k-1)$ for $n\in\Z$ (the eigenvalues of $\sD$). By Corollary \ref{muntrelmunk:cor} we have
\begin{equation}
\label{estnun:eq}
e^{-\ca}\abs{\mu_n(k)}\le\abs{\nu_n}\le e^{\ca}\abs{\mu_n(k)}.
\end{equation}
Choose an orthonormal basis $\{\xi_n:n\in\Z\}$ of $\spin$ with $\sD\xi_n=\nu_n\xi_n$. 
For each $n\in\Z$ set $\kappa^\pm_n=\kappa^\pm(\nu_n)$ and $u^\pm_n=\xi_n\otimes x^\pm_{\nu_n}\in\bspin$.
The definitions of $\oK_0$ and $X_d$ lead to $\oK_0u^\pm_n=\kappa^\pm_n u^\pm_n$
so, in particular, $\{u^+_n,u^-_n:n\in\Z\}$ is an eigenbasis for $\oK_0$.

Given $\epsilon,R>0$ set $\Zazm=\{n\in\Z:\abs{\mu_n(k)}\le\epsilon\}$ and $\Zbig=\{n\in\Z:\abs{\mu_n(k)}>R\}$.
Let $\PKazm$, $\PKbig$ and $\PKrem$ denote the (orthogonal) spectral projections of $\oK_0$ with
\[
\Ran\PKazm=\Span\{u^\pm_n\,:\,n\in\Zazm\},\quad
\Ran\PKbig=\Span\{u^+_n,u^-_n\,:\,n\in\Zbig\}
\]
and $\PKrem=\oI-\PKazm[+]-\PKazm[-]-\PKbig$. 
Clearly
\[
\dim\Ran\PKazm=\#\Zazm=\cfD{\epsilon}
\quad\text{and}\quad
\dim\Ran\PKrem=2\dfD{\epsilon,R}.
\]

\begin{lem}
\label{estK0:lem}
Let $\epsilon,R>0$. Then $\pm\oK_0\PKazm\ge0$ while $\bignorm{[\oK_0-\kappa^\pm(0)\vtsp\oI]\vtsp\PKazm}\le\Co{2,1}\epsilon^2$
and $\norm{\oK_0\PKbig}\le\Co{2,2}R^{-1}$ for some constants $\Co{2,1}$ and $\Co{2,2}$.
\end{lem}

\begin{proof}
We have $\pm\kappa^\pm_n>0$ for all $n$ while Lemma \ref{behkappa:lem} and \eqref{estnun:eq} give
\[
\abs{\kappa^\pm_n-\kappa^\pm(0)}\le16\,\nu_n^2\le16e^{2\ca}\epsilon^2
\]
for $n\in\Zazm$, and 
$\abs{\kappa^\pm_n}\le\abs{\nu_n}^{-1}\le e^{\ca}R^{-1}$ for $n\in\Zbig$.
The result follows (with $\Co{2,1}=16e^{2\ca}$ and $\Co{2,2}=e^{\ca}$).
\end{proof}

Next we consider $\oK_1$; we begin with estimates for $\oK_1$ restricted to certain spectral subspaces of $\oK_0$. 

\begin{lem}
\label{estK1:lem}
Suppose $\epsilon,R>0$ and $\pi_1,\pi_2\in\{+,-\}$. 
Then $\norm{\PKazm[\pi_1]\oK_1\PKazm[\pi_2]}\le\Co{3,1}\vtsp\epsilon\vtsp\cfD{\epsilon}$ and 
$\norm{\oK_1\PKbig}\le\Co{3,2}R^{-1/2}$ for some constants $\Co{3,1}$ and $\Co{3,2}$.
\end{lem}

\begin{proof}
Since $\{u_n^\pm:n\in\Zazm\}$ is an orthonormal basis for $\Ran\PKazm$ we have
\begin{equation}
\label{HSestofNPK1P:eq}
\norm{\PKazm[\pi_1]\oK_1\PKazm[\pi_2]}^2\le\sum_{m,n\in\Zazm}\abs{\ipd{u^{\pi_1}_m}{\oK_1u^{\pi_2}_n}}^2.
\end{equation}
Now the definitions of $\oK_1$ and $\oQ_1$ give
\begin{align*}
\ipd{u^{\pi_1}_m}{\oK_1u^{\pi_2}_n}
&=\bigipd{\abs{\oP}^{-1/2}u^{\pi_1}_m}{\oQ_1\abs{\oP}^{-1/2}u^{\pi_2}_n}\\
&=2\ipd{\xi_m}{\sA\xi_n}\,
\bigipd{\abs{X_{\nu_m}}^{-1/2}x^{\pi_1}_{\nu_m}}{\sigma_3\abs{X_{\nu_n}}^{-1/2}x^{\pi_2}_{\nu_n}}.
\end{align*}
Note that $\norm{\sigma_3}=1$ so
$\abs{\ipd{\abs{X_{\nu_m}}^{-1/2}x^{\pi_1}_{\nu_m}}{\sigma_3\abs{X_{\nu_n}}^{-1/2}x^{\pi_2}_{\nu_n}}}\le2$
by \eqref{normXdiest:eq}.
On the other hand when $m,n\in\Zazm$ Proposition \ref{spinalign1:prop} and \eqref{estnun:eq} give
\[
\abs{\ipd{\xi_m}{\sA\xi_n}}=\abs{\ipd{\xi_m}{\sigma(\alpha)\xi_n}}
\le\pi(\abs{\nu_m}+\abs{\nu_n})\vtsp\norm{\alpha}_{L^\infty}
\le\ca e^{\ca}\epsilon.
\]
Therefore $\abs{\ipd{u^{\pi_1}_m}{\oK_1u^{\pi_2}_n}}\le\Co{3,1}\epsilon$ with $\Co{3,1}=4\ca e^{\ca}$.
Since $\#\Zazm=\cfD{\epsilon}$ the first part of the result now follows from \eqref{HSestofNPK1P:eq}.

Now let $u\in\Ran\PKbig$. Then $u=\sum_{n\in\Zbig}\xi_n\otimes z_n$ for some $z_n\in\C^2$, so
\[
\abs{\oP}^{-1/2}u=\sum_{n\in\Zbig}\xi_n\otimes\abs{X_{\nu_n}}^{-1/2}z_n.
\]
For $n\in\Zbig$ \eqref{normXdiest:eq} and \eqref{estnun:eq} lead to
\[
\norm{\abs{X_{\nu_n}}^{-1/2}z_n}^2
\le\abs{\nu_n}^{-1}\norm{z_n}^2
\le e^{\ca}R^{-1}\norm{z_n}^2.
\]
Since $\{\xi_n:n\in\Zbig\}$ is an orthonormal set (in $\spin$) it follows that
\[
\norm{\abs{\oP}^{-1/2}u}^2
=\sum_{n\in\Zbig}\norm{\abs{X_{\nu_n}}^{-1/2}z_n}^2
\le e^{\ca}R^{-1}\sum_{n\in\Zbig}\norm{z_n}^2
=e^{\ca}R^{-1}\norm{u}^2.
\]
Therefore $\norm{\abs{\oP}^{-1/2}\PKbig}\le e^{\ca/2}R^{-1/2}$. Since $\norm{\abs{\oP}^{-1/2}}\le\sqrt{2}$ and $\norm{\oQ_1}=2\norm{\sA}=2\norm{\alpha}_{L^\infty}$ the required estimate for $\norm{\oK_1\PKbig}$ follows with $\Co{3,2}=2\sqrt2 e^{\ca/2}\norm{\alpha}_{L^\infty}$.
\end{proof}

For $\epsilon,R>0$ set 
\begin{equation}
\label{defdeltaeR:eq}
\delta(\epsilon,R)=\Co{2,1}\epsilon^2+4\Co{3,1}\epsilon\vtsp\cfD{\epsilon}+\Co{2,2}R^{-1}+2\Co{3,2}R^{-1/2}.
\end{equation}
Let $\{\lambda_n^+:n\in\N\}$ and $\{\lambda_n^-:n\in\N\}$ denote the sets of positive and negative eigenvalues of $\oK=\oK_0+\oK_1$, enumerated to include multiplicities and ordered so that $\lambda_1^-\le\lambda_2^-\le\dots<0<\dots\le\lambda_2^+\le\lambda_1^+$. 

\begin{prop}
\label{minmaxappK:prop}
Suppose $\epsilon,R>0$. Then
\begin{equation}
\label{outestK1eval:eq}
\#\bigl\{n\in\N\vtsp:\vtsp\abs{\lambda_n^\pm}>\pm\kappa^\pm(0)+\delta(\epsilon,R)\bigr\}\le2\dfD{\epsilon,R}
\end{equation}
and
\begin{equation}
\label{inestK1eval:eq}
\#\bigl\{n\in\N:\abs{\lambda_n^\pm}>\delta(\epsilon,R)\bigr\}\le\cfD{\epsilon}+2\dfD{\epsilon,R}.
\end{equation}
\end{prop}

The basic argument is a variational one with Lemmas \ref{estK0:lem} and \ref{estK1:lem} providing the relevant information about $\oK_0$ and $\oK_1$ respectively.

\begin{proof}
Set $M=\dim\Ran\PKrem=2\dfD{\epsilon,R}$ and let $H\le\bspin$ with $\dim H=M+1$. Choose $u\in H$ with $\norm{u}=1$ and $\PKrem u=0$. Then $u=(\PKazm[+]+\PKazm[-]+\PKbig)u$ so
\[
\ipd{u}{\oK_0u}=\ipd{u}{\oK_0\PKazm[+]u}+\ipd{u}{\oK_0\PKazm[-]u}+\ipd{u}{\oK_0\PKbig u}
\le\norm{\oK_0\PKazm[+]}+\norm{\oK_0\PKbig}
\]
(since $\oK_0\PKazm[-]\le0$ from Lemma \ref{estK0:lem}) while
\begin{align*}
\ipd{u}{\oK_1u}
&=\ipd{u}{(\PKazm[+]+\PKazm[-])\oK_1(\PKazm[+]+\PKazm[-])u}+\ipd{u}{\oK_1\PKbig u}+\ipd{\oK_1\PKbig u}{(\PKazm[+]+\PKazm[-])u}\\
&\le\norm{(\PKazm[+]+\PKazm[-])\oK_1(\PKazm[+]+\PKazm[-])}+2\norm{\oK_1\PKbig}.
\end{align*}
Using a variational argument it follows that
\[
\lambda_{M+1}^+\le\norm{\oK_0\PKazm[+]}+\norm{\oK_0\PKbig}+\norm{(\PKazm[+]+\PKazm[-])\oK_1(\PKazm[+]+\PKazm[-])}+2\norm{\oK_1\PKbig}.
\]
Lemmas \ref{estK0:lem} and \ref{estK1:lem} then give $\lambda_{M+1}^+\le\kappa^+(0)+\delta(\epsilon,R)$. 
The case of the upper sign in \eqref{outestK1eval:eq} clearly follows. The lower sign can be obtained by a similar argument.

Now set $M=\dim\Ran(\PKazm[+]+\PKrem)=\cfD{\epsilon}+2\dfD{\epsilon,R}$. A slightly simpler version of the above argument leads to
\[
\lambda_{M+1}^+\le\norm{\oK_0\PKbig}+\norm{\PKazm[-]\oK_1\PKazm[-]}+2\norm{\oK_1\PKbig}
\le\Co{2,2}R^{-1}+\Co{3,1}\epsilon\vtsp\cfD{\epsilon}+2\Co{3,2}R^{-1/2}.
\]
Since the right hand side is clearly bounded above by $\delta(\epsilon,R)$ \eqref{inestK1eval:eq} now follows.
\end{proof}

\section{The upper bound}

\label{mainarg:sec}

The upper bound in Theorem \ref{typeIIest1:thm} follows from Theorem \ref{basicappzmS2:thm} if we can show that 
$\cfnk-\cfD[k,\alpha]{\epsilon_k}$ is bounded from above by $o(\abs{k})$ for suitably chosen $\epsilon_k$.
We firstly estimate this difference using Propositions \ref{evalLinfo:prop} and \ref{minmaxappK:prop} together with Weyl's inequality. 

\begin{lem}
\label{estN0genform:lem}
Suppose $0<\epsilon\le1/\Co{1}$ and $R>0$. Then
\[
\bigl[\tfrac23-\delta(\epsilon,R)\bigr]\bigl[\cfnk-\cfD[k,\alpha]{\epsilon}\bigr]
\le[\delta(\epsilon,R)+\Co{1}\epsilon]\vtsp\cfD[k,\alpha]{\epsilon}+2\norm{\oK}\vtsp\dfD[k,\alpha]{\epsilon,R}.
\]
\end{lem}

\begin{proof}
Put $N=2\cfnk$ and $M=N-2\cfD{\epsilon}$.
Let $\Lambda=\bigl\{\lambda\in\spec(\oL):\lambda^{-1}\in\gaI\bigr\}$ so $\#\Lambda=N$ by Proposition \ref{evalLinfo:prop}.
Also let $\mathrm{K}\subset\spec(\oK)$ denote the collection of the $N$ eigenvalues of $\oK$ with largest moduli. 
Since $\oL^*\oL=\oK^2$ the singular values of $\oL$ are precisely the moduli of the eigenvalues of $\oK$.
Weyl's inequality (\cite{W}) then gives
\begin{equation}
\label{WeylIneq:eq}
\sum_{\lambda\in\Lambda}\lambda\le\sum_{\lambda\in\mathrm{K}}\vtsp\abs{\lambda}.
\end{equation}

For any $\lambda\in\Lambda$ we have $\lambda^{-1}\in\gaI=[1/2,3/2]$ so $\lambda\ge2/3$. If $\lambda^{-1}\in\gaIe{-}$ then
\[
\lambda\ge\frac1{\gae{-}}=\frac2{1+\Co{1}\epsilon}\ge2\vtsp(1-\Co{1}\epsilon).
\]
From Proposition \ref{evalLinfo:prop} it follows that
\begin{equation}
\label{estSLeval:eq}
\sum_{\lambda\in\Lambda}\lambda\ge2\vtsp(1-\Co{1}\epsilon)\vtsp\cfD{\epsilon}+\tfrac23[N-\cfD{\epsilon}]
=2\bigl(\vtsp\tfrac43-\Co{1}\epsilon\bigr)\cfD{\epsilon}+\tfrac23 M.
\end{equation}

Write $\delta=\delta(\epsilon,R)$. 
Proposition \ref{minmaxappK:prop} shows that $\oK$ has at most $2\dfD{\epsilon,R}$ eigenvalues in 
each of the intervals $(-\infty,-2/3-\delta)$ and $(2+\delta,\infty)$, 
and at most $\cfD{\epsilon}+2\dfD{\epsilon,R}$ eigenvalues in each of the intervals $(-\infty,-\delta)$ and $(\delta,\infty)$.
Furthermore, the spectral radius of $\oK$ is $\norm{\oK}$ while
$\#\mathrm{K}-(2\cfD{\epsilon}+4\dfD{\epsilon,R})=M-4\dfD{\epsilon,R}\le M$.
Therefore
\begin{equation}
\label{estSKeval:eq}
\sum_{\lambda\in\mathrm{K}}\vtsp\abs{\lambda}
\le4\norm{\oK}\vtsp\dfD{\epsilon,R}+\bigl(\tfrac23+\delta\bigr)\cfD{\epsilon}+(2+\delta)\vtsp\cfD{\epsilon}+\delta M.
\end{equation}
The result now follows when we combine \eqref{WeylIneq:eq}, \eqref{estSLeval:eq} and \eqref{estSKeval:eq}.
\end{proof}

\begin{rems}
Key to our argument is the identification of those eigenvalues and singular values of $\oL$ which arise from the Aharonov-Casher and approximate zero modes. These contribute $\tfrac83\cfD{\epsilon}$ to each side of \eqref{WeylIneq:eq}, the cancellation of which
allows the quantity $\cfnk-\cfD[k,\alpha]{\epsilon}$ to be estimated with sufficient precision.
\end{rems}

Since $\norm{\abs{\oP}^{-1/2}}\le\sqrt2$ straightforward bounds on $\oQ$ give
\begin{equation}
\label{normbndK:eq}
\norm{\oK}\le2\norm{\oQ_0+\oQ_1}=2(1+4\norm{\alpha}_{L^\infty}^2)^{1/2}.
\end{equation}

\begin{proof}[Proof of upper bound in Theorem \ref{typeIIest1:thm}]
Set $\epsilon_k=e^{-\abs{k}^{1/2}}$ and $R_k=\abs{k}^{1/4}$ for all $k\in\Z$.
As $\abs{k}\to\infty$ we clearly have $\epsilon_k=o(\abs{k}^{-1})$ and $R_k\to\infty$, while 
Theorem \ref{basicappzmS2:thm} gives $\cfD{\epsilon_k}=\flux{\abs{\beta}}\vtsp\abs{k}+o(\abs{k})$ and $\dfD{\epsilon_k,R_k}=o(\abs{k})$.
It follows that $\delta(\epsilon_k,R_k)=o(1)$ (recall \eqref{defdeltaeR:eq}) and so $\cfnk-\cfD[k,\alpha]{\epsilon_k}=o(\abs{k})$ by Lemma \ref{estN0genform:lem} and \eqref{normbndK:eq}. 
\end{proof}

\section{Approximate zero modes on $\sphere$}

\label{S2ctfn:sec}

Let $\sphere[+]$ (respectively $\sphere[-]$) denote the sphere with the south (respectively north) pole removed; if we view $\sphere$ as the unit sphere in $\R^3$ then $\sphere[\pm]=\sphere\setminus\{(0,0,\mp1)\}$. Let $\stproj{\pm}:\sphere[\pm]\to\R^2$ denote stereographic projection, given by
\[
\stproj{\pm}(x)=\frac1{1\pm x_3}(x_1,x_2),
\quad x=(x_1,x_2,x_3)\in\sphere[\pm].
\]
Set $\confwR(x)=2(1+\abs{x}^2)^{-1}$ for $x\in\R^2$, and $\confwS{\pm}=\confwR\circ\stproj{\pm}$.
It is straightforward to check that the map $\stproj{\pm}$ is an isometry if $\R^2$ is given the conformal metric $\confwR\vtsp\fipd{\R^2}{\cdot}{\cdot}$ (where $\fipd{\R^2}{\cdot}{\cdot}$ is the usual Euclidean metric on $\R^2$). 
Hence $\stproj{\pm}^*(\confwR^2\vtsp\volR)=\volS$ (where $\volR=dx_1\wedge dx_2$ is the usual volume form on $\R^2$).

For any $\delta\in[0,1]$ set $\sphere[\delta,\pm]=\sphere\cap\{\pm x_3<\delta\}$; in particular $\sphere[1,\pm]=\sphere[\pm]$ while $\sphere[0,+]$ and $\sphere[0,-]$ are the north and south hemispheres. It is easy to check that $\stproj{\pm}(\sphere[\delta,\pm])=\disc[r_\delta]$ where $r_\delta^2=(1+\delta)/(1-\delta)$,
while we have the bounds 
\begin{equation}
\label{confwRbnd:eq}
1-\delta<\confwR(x)\le 2,
\quad x\in\disc[r_\delta].
\end{equation}

\medskip

Using the isometry $\stproj{\pm}^{-1}$ we can pull-back the (restricted) spin$^c$ bundle $\spbun{k}$ from $\sphere[\pm]$ to get a spin$^c$ bundle on $\R^2$. 
Since $\R^2$ is contractible the latter is isomorphic to the trivial bundle $\R^2\times\C^2$, so sections of this bundle (spinors) can be identified with maps $\R^2\to\C^2$. 
For $\xi\in\Gamma(\spbun{k})$ with $\supp(\xi)\subset\sphere[\pm]$ let $\eta=\xi\circ\stproj{\pm}^{-1}$ denote the corresponding map in $\cont[0]{\infty}(\R^2,\C^2)$. Then
\begin{equation}
\label{normrelSR:eq}
\norm{\xi}_{L^2(\sphere)}^2
=\int_{\sphere[\pm]}\abs{\xi}_{\spbun{k}}^2\volS
=\int_{\R^2}\abs{\xi\circ\stproj{\pm}}^2\,\confwR^2\vtsp\volR
=\norm{\confwR\eta}_{L^2(\R^2)}^2.
\end{equation}
Using the isometry $\stproj{\pm}$ and the above identification of spin$^{c}$ bundles any Dirac operator on $\sphere$ can be restricted to $\sphere[\pm]$ and then considered as a Dirac operator on $\R^2$ with the conformal metric $\confwR\vtsp\fipd{\R^2}{\cdot}{\cdot}$. Conformal mapping properties of Dirac operators (see \cite[Section 1.4]{H} or \cite[Theorem 4.3]{ES}) mean the latter is simply related to a Dirac operator on $\R^2$ with the usual metric. 
Under the above identification of spin$^{c}$ bundles a Dirac operator on $\R^2$ becomes a Weyl-Dirac operator corresponding to a potential $A'=A'_1\d x_1+A'_2\d x_2$ on $\R^2$; that is, an operator given by the $2$-dimensional version of \eqref{WDR3:eq}. 
More precisely let $\alpha\in\formS{1}$ and consider the Dirac operator $\sDirac{k}{\alpha}$ on $\spbun{k}$. Then we can find $A_{\pm}\in\Omega^1(\R^2)$ so that
\begin{equation}
\label{DconfSR:eq}
\bigl(\confwS{\pm}^{3/2}\sDirac{k}{\alpha}\confwS{\pm}^{-1/2}\bigr)(\eta\circ\stproj{\pm})
=(\Dirac{\R^2,A_\pm}\eta)\circ\stproj{\pm}
\end{equation}
for all $\eta:\R^2\to\C^2$ (note that $\eta\circ\stproj{\pm}\in\Gamma(\spbun{k}_{\pm})$, where $\spbun{k}_{\pm}$ is the restriction of $\spbun{k}$ to $\sphere[\pm]$). 
Furthermore the magnetic field corresponding to $\Dirac{\R^2,A_\pm}$ is simply the pull-back of that corresponding to $\sDirac{k}{\alpha}$ under the map $\stproj{\pm}^{-1}$; 
if the latter is $\beta=f\volS$ then the former will be given by $\beta_\pm=dA_\pm=(f\circ\stproj{\pm}^{-1})\,\confwR^2\vtsp\volR$.
In particular for any open subset $U\subseteq\R^2$ we have
\begin{equation}
\label{BSRrel1:eq}
\int_{U}\beta_\pm=\int_{\stproj{\pm}^{-1}(U)}\beta.
\end{equation}

\smallskip

For $A'\in\Omega^1(\R^2)$ and $r>0$ let $\Pauli_{\disc[r],A'}$ denote the 
Pauli operator on $\disc[r]$ with magnetic potential $A'$ and Dirichlet boundary conditions;
this can be defined as the non-negative self-adjoint operator associated to the closure of the quadratic form given by $\eta\mapsto\norm{\Dirac{\R^2,A'}\eta}_{L^2(\R^2)}^2$ for $\eta\in\cont[0]{\infty}(\disc[r],\C^2)$.

\smallskip

For the next result let $\sDirac{k}{\alpha}$ denote a Dirac operator on $\spbun{k}$ and let $A_{\pm}$ denote the corresponding $1$-forms on $\R^2$ as discussed above.

\begin{prop}
\label{cntfncorrs:prop}
There exists $\Co{4}>0$ so that
for any $\mu>0$ and $\delta\in(0,1]$ we have 
\begin{align}
&\#\{\lambda\in\spec(\sDirac{k}{\alpha}):\abs{\lambda}\le\mu\}\nonumber\\
\label{cntfncorrslow:eq}
&\qquad{}\ge\#\bigl\{\lambda\in\spec(\Pauli_{\disc,A_+}):\lambda\le\mu^2\bigr\}
+\#\bigl\{\lambda\in\spec(\Pauli_{\disc,A_-}):\lambda\le\mu^2\bigr\}
\end{align}
and
\begin{align}
&\#\bigl\{\lambda\in\spec(\sDirac{k}{\alpha}):\lambda^2\le\mu^2-\Co{4}\delta^{-2}\bigr\}\nonumber\\
\label{cntfncorrsup:eq}
&\ \ {}\le\#\bigl\{\lambda\in\spec(\Pauli_{\disc[r_\delta],A_+}):\lambda\le(4\mu)^2\bigr\}
+\#\bigl\{\lambda\in\spec(\Pauli_{\disc[r_\delta],A_-}):\lambda\le(4\mu)^2\bigr\}.
\end{align}
\end{prop}

\begin{proof}
Let $\eta_\pm\in\cont[0]{\infty}(\disc,\C^2)$. Set $\xi_{\pm}=(\confwR^{-1/2}\eta_{\pm})\circ\stproj{\pm}^{-1}$ giving $\xi_\pm\in\Gamma(\spbun{k}_\pm)$ with $\supp(\xi_\pm)\subseteq\sphere[0,\pm]$. 
Extend $\xi_\pm$ by $0$ and set $\xi=\xi_++\xi_-\in\Gamma(\spbun{k})$. 
From \eqref{confwRbnd:eq} we have $\confwR\ge1$ on $\disc$. Together with \eqref{normrelSR:eq} and \eqref{DconfSR:eq} we then get
\[
\norm{\xi_{\pm}}_{L^2(\sphere[\pm])}^2=\bignorm{\confwR^{1/2}\eta_{\pm}}_{L^2(\disc)}^2
\ge\norm{\eta_{\pm}}_{L^2(\disc)}^2
\]
and
\[
\norm{\sDirac{k}{\alpha}\xi_{\pm}}_{L^2(\sphere[\pm])}^2=\bignorm{\confwR^{-1/2}\Dirac{\R^2,A_{\pm}}\eta_{\pm}}_{L^2(\disc)}^2
\le\norm{\Dirac{\disc,A_{\pm}}\eta_{\pm}}_{L^2(\disc)}^2.
\]
Since $\xi_+$ and $\xi_-$ have disjoint support it follows that
\[
\norm{\xi}_{L^2(\sphere)}^2
=\norm{\xi_{+}}_{L^2(\sphere[+])}^2+\norm{\xi_{-}}_{L^2(\sphere[-])}^2
\ge\norm{\eta_{+}}_{L^2(\disc)}^2+\norm{\eta_{-}}_{L^2(\disc)}^2
\]
and
\[
\norm{\sDirac{k}{\alpha}\xi}_{L^2(\sphere)}^2
=\norm{\sDirac{k}{\alpha}\xi_+}_{L^2(\sphere[+])}^2+\norm{\sDirac{k}{\alpha}\xi_-}_{L^2(\sphere[-])}^2
\le\norm{\Dirac{\disc,A_{+}}\eta_{+}}_{L^2(\disc)}^2+\norm{\Dirac{\disc,A_{-}}\eta_{-}}_{L^2(\disc)}^2.
\]
A standard variational argument then leads to \eqref{cntfncorrslow:eq}.

\medskip

Now choose non-negative functions $\puS{\pm}\in\cont[0]{\infty}(\sphere[\delta,\pm])$ so that $\puS{+}^2+\puS{-}^2=1$ and $\abs{d\puS{\pm}}\le\Co{4,0}\delta^{-1}$ on $\sphere$, where $\Co{4,0}$ is independent of $\delta$.
Let $\xi\in\Gamma(\spbun{k})$ and define compactly supported sections of $\spbun{k}_\pm$ by setting $\pxiS{\pm}=\puS{\pm}\xi$.
Also set $\petaR{\pm}=\confwR^{1/2}\pxiS{\pm}\circ\stproj{\pm}$ giving $\petaR{\pm}\in\cont[0]{\infty}(\disc[r_\delta],\C^2)$.
Then \eqref{normrelSR:eq}, (the upper bound in) \eqref{confwRbnd:eq} and \eqref{DconfSR:eq} give
\[
\norm{\pxiS{\pm}}_{L^2(\sphere[\pm])}^2
=\bignorm{\confwR^{1/2}\petaR{\pm}}_{L^2(\R^2)}^2
\le2\norm{\petaR{\pm}}_{L^2(\disc[r_\delta])}^2
\]
so
\[
\norm{\xi}_{L^2(\sphere)}^2
=\norm{\pxiS{+}}_{L^2(\sphere[+])}^2+\norm{\pxiS{-}}_{L^2(\sphere[-])}^2
\le2\bigl[\norm{\petaR{+}}_{L^2(\disc[r_\delta])}^2+\norm{\petaR{-}}_{L^2(\disc[r_\delta])}^2\bigr].
\]
Similarly
\[
\norm{\sDirac{k}{\alpha}\pxiS{\pm}}_{L^2(\sphere[\pm])}^2=\bignorm{\confwR^{-1/2}\Dirac{\R^2,A_{\pm}}\petaR{\pm}}_{L^2(\R^2)}^2
\ge\tfrac12\vtsp\norm{\Dirac{\disc[r_\delta],A_{\pm}}\petaR{\pm}}_{L^2(\disc[r_\delta])}^2
\]
while
\begin{align*}
&\norm{\sDirac{k}{\alpha}\xi}_{L^2(\sphere)}^2
=\norm{\puS{+}\sDirac{k}{\alpha}\xi}_{L^2(\sphere)}^2+\norm{\puS{-}\sDirac{k}{\alpha}\xi}_{L^2(\sphere)}^2\\
&\qquad{}=\bignorm{\sDirac{k}{\alpha}\pxiS{+}-i\sigma(d\puS{+})\xi}_{L^2(\sphere)}^2+\bignorm{\sDirac{k}{\alpha}\pxiS{-}-i\sigma(d\puS{-})\xi}_{L^2(\sphere)}^2\\
&\qquad{}\ge\tfrac12\bigl[\norm{\sDirac{k}{\alpha}\pxiS{+}}_{L^2(\sphere[+])}^2+\norm{\sDirac{k}{\alpha}\pxiS{-}}_{L^2(\sphere[-])}^2\bigr]
-2\Co{4,0}^2\delta^{-2}\norm{\xi}_{L^2(\sphere)}^2.
\end{align*}
Therefore
\[
\norm{\sDirac{k}{\alpha}\xi}_{L^2(\sphere)}^2+2\Co{4,0}^2\delta^{-2}\norm{\xi}_{L^2(\sphere)}^2
\ge\tfrac12\bigl[\norm{\Dirac{\disc[r_\delta],A_{+}}\petaR{+}}_{L^2(\disc[r_\delta])}^2+
\norm{\Dirac{\disc[r_\delta],A_{-}}\petaR{-}}_{L^2(\disc[r_\delta])}^2\bigr].
\]
A standard variational argument now gives \eqref{cntfncorrsup:eq} (with $\Co{4}=2\Co{4,0}^2$; note that $\Gamma(\spbun{k})$ is a core for $\sDirac{k}{\alpha}$).
\end{proof}

We can use \eqref{DconfSR:eq} to transfer results about approximate zero modes on $\R^2$ to $\sphere$; information about the former was obtained in \cite{E3}.

\begin{proof}[Proof of Theorem \ref{basicappzmS2:thm}]
For each $k\in\Z$ we have a Dirac operator $\sDirac{k}{k\alpha}$ on $\spbun{k}$ with magnetic $2$-form $k\bigl(\frac12\volS+d\alpha\bigr)$. Pulling this back to $\R^2$ using $\stproj{\pm}$ as discussed above, 
we can arrange so that the corresponding $1$-forms on $\R^2$ are simply $kA_{\pm}$ for fixed ($k$ independent) $1$-forms $A_{\pm}$. 
The corresponding field is $k\beta_{\pm}$ where $\beta_\pm=d A_{\pm}$. By \eqref{BSRrel1:eq} we have
\begin{equation}
\label{BSRrel2:eq}
\int_{\disc[r_\delta]}\abs{\beta_{\pm}}=\int_{\sphere[\delta,\pm]}\abs{\beta},
\qquad\delta\in[0,1].
\end{equation}

\medskip

From \cite[Theorem 1.2]{E3} and \eqref{BSRrel2:eq} we get
\[
\liminf_{\abs{k}\to\infty}\frac{1}{\abs{k}}\,\#\bigl\{\lambda\in\spec(\Pauli_{\disc,kA_{\pm}}):\lambda\le\epsilon_k^2\bigr\}
\ge\frac1{2\pi}\int_{\disc}\abs{\beta_{\pm}}
=\frac1{2\pi}\int_{\sphere[0,\pm]}\abs{\beta}.
\]
Combined with Proposition \ref{cntfncorrs:prop} we then have
\[
\liminf_{\abs{k}\to\infty}\frac{1}{\abs{k}}\,\cfD[k,\alpha]{\epsilon_k}
\ge\frac1{2\pi}\int_{\sphere[0,+]}\abs{\beta}\,+\,\frac1{2\pi}\int_{\sphere[0,-]}\abs{\beta}
=\flux{\abs{\beta}}.
\]

\medskip

Now let $\delta>0$ and set $\widetilde{R}_k=16(R_k^2+\Co{4}\delta^{-2})$ for $k\in\Z$. 
Then $\widetilde{R}_k=o(\abs{k})$ as $\abs{k}\to\infty$, so \cite[Theorem 1.1]{E3} and \eqref{BSRrel2:eq} give
\[
\limsup_{\abs{k}\to\infty}\frac{1}{\abs{k}}\,\#\bigl\{\lambda\in\spec(\Pauli_{\disc[r_\delta],kA_{\pm}}):\lambda\le\widetilde{R}_k\bigr\}
\le\frac1{2\pi}\int_{\disc[r_\delta]}\abs{\beta_{\pm}}
=\frac1{2\pi}\int_{\sphere[\delta,\pm]}\abs{\beta}.
\]
Combined with Proposition \ref{cntfncorrs:prop} we then have 
\[
\limsup_{\abs{k}\to\infty}\frac{1}{\abs{k}}\,\cfD[k,\alpha]{R_k}
\le\frac1{2\pi}\int_{\sphere[\delta,+]}\abs{\beta}\,+\,\frac1{2\pi}\int_{\sphere[\delta,-]}\abs{\beta}
=\flux{\abs{\beta}}+O(\delta)
\]
as $\delta\to0^+$ (note that $\beta$ is bounded while $\abs{\sphere[\delta,+]\cap\sphere[\delta,-]}=O(\delta)$). 
Taking $\delta\to0^+$ leads to the stated upper bound for $\cfD[k,\alpha]{R_k}$.
\end{proof}

\section{Spin-field estimates on $\sphere$}

\label{S2spfldest:sec}

For any $n$ let $d:\formS{n}\to\formS{n-1}$ and $\delta:\formS{n}\to\formS{n-1}$ denote the exterior derivative and its adjoint with respect to the Hodge $*$ operator. We have $*:\formS{n}\to\formS{2-n}$ with $**=(-1)^n$ and $\delta=-*d\,*$. Also $*\volS=1$.

The expression $d\delta+\delta d$ defines the Laplace-de Rham operator on $n$-forms.
For $n=0$ this reduces to $\delta d=-\Delta$, the negative of the Laplace-Beltrami operator on (scalar) functions.
The Green's function for the latter is given in terms of $\log(1-x.y)$ (where the dot product is defined by viewing $\sphere$ as the unit sphere in $\R^3$);
more precisely for any $f\in\cont{\infty}(\sphere)$ with $\int_{\sphere} f\volS=0$ we have
\begin{equation}
\label{ssintrep:eq}
f(x)=\frac1{4\pi}\int_{\sphere}\log(1-x.y)\,\Delta f(y)\,\volS(y)
\end{equation}
for all $x\in\sphere$ (see \cite[Theorem 4.15]{FS}). From this we can obtain a related integral representation for $1$-forms. Firstly for any $y\in\R^3$ let $\iof{y}\in\formS{1}$ denote the exterior derivative of $x\mapsto x.y$. 

\begin{prop}
\label{s1fintrep:prop}
For any $\omega\in\formS{1}$ and $x\in\sphere$ we have
\[
\omega(x)
=\frac1{4\pi}\int_{\sphere}\frac{\iof{y}(x)}{1-x.y}\,\delta\omega(y)\,\volS(y)
-\frac1{4\pi}\int_{\sphere}\frac{(*\iof{y})(x)}{1-x.y}\,d\omega(y).
\]
\end{prop}

\begin{proof}
Suppose $f\in\cont{\infty}(\sphere)$ satisfies $\int_{\sphere}f\volS=0$. Taking the exterior derivative of \eqref{ssintrep:eq} with respect to $x$ gives
\begin{equation}
\label{sdsintrep:eq}
df(x)=\frac1{4\pi}\int_{\sphere}\frac{\iof{y}(x)}{1-x.y}\,\delta df(y)\,\volS(y).
\end{equation}
Now suppose $\nu\in\formS{2}$ with $\int_{\sphere}\nu=0$. Set $g=*\nu\in\cont{\infty}(\sphere)$ so $\delta\nu=-*dg$ and $(\delta dg)\volS=d\delta\nu$. Applying the Hodge $*$ to \eqref{sdsintrep:eq} then leads to
\begin{equation}
\label{sdeltasintrep:eq}
\delta\nu(x)=-\frac1{4\pi}\int_{\sphere}\frac{(*\iof{y})(x)}{1-x.y}\,d\delta\nu(y).
\end{equation}
Finally suppose $\omega\in\formS{1}$. Since $H^1(\sphere)=0$ the Hodge decomposition theorem gives $f\in\cont{\infty}(\sphere)$ and $\nu\in\formS{2}$ such that $\omega=df+\delta\nu$. Since $d1=0=\delta\volS$ we may assume $\int_{\sphere}f\,\volS=0=\int_{\sphere}\nu$. The result now follows from \eqref{sdsintrep:eq} and \eqref{sdeltasintrep:eq}.
\end{proof}

For any $x,y\in\sphere$ it is easy to check $\abs{\iof{y}(x)}_{\sphere}=\abs{(*\iof{y})(x)}_{\sphere}=1-(x.y)^2$. A straightforward calculation then gives
\[
\int_{\sphere}\lrabs{\frac{\iof{y}(x)}{1-x.y}}_{\sphere}\!\volS(x)=2\pi^2=\int_{\sphere}\lrabs{\frac{(*\iof{y})(x)}{1-x.y}}_{\sphere}\!\volS(x).
\]
Coupled with Proposition \ref{s1fintrep:prop} we immediately get the following estimate for $1$-forms.

\begin{cor}
\label{1formL1est:cor}
For any $\omega\in\formS{1}$ we have $\norm{\omega}_{L^1}\le\tfrac12\pi\vtsp(\norm{\delta\omega}_{L^1}+\norm{d\omega}_{L^1})$.
\end{cor}

When needed $\{\vb{1},\vb{2}\}$ denotes an orthonormal frame (of local vector fields) while $\{\fb{1},\fb{2}\}$ denotes the corresponding orthonormal dual frame (of local $1$-forms). 
We assume $\{\vb{1},\vb{2}\}$ is positively oriented so $\volS=\fb{1}\wedge\fb{2}$. Also $*\fb{1}=\fb{2}$ and $*\fb{2}=-\fb{1}$. For any $\omega\in\formS{1}$ we have the local expression
\begin{equation}
\label{divconn:eq}
\delta\omega=-\trace\LCc\omega=-\bigl[(\LCc[\vb{1}]\omega)(\vb{1})+(\LCc[\vb{2}]\omega)(\vb{2})\bigr],
\end{equation}
where $\LCc$ denotes the Levi-Civita connection (on $1$-forms; see \cite[Lemma 4.8]{GHL}).

\medskip

For any spinors $\xi,\eta\in\Gamma(\spbun{k})$ let $\spf{\xi}{\eta}\in\formS{1}$ be the unique $1$-form satisfying
\[
\fipd{\sphere}{\spf{\xi}{\eta}}{\rho}=\fipd{\spbun{k}}{\xi}{\sigma(\rho)\eta}
\]
for all $\rho\in\formS{1}$. In terms of a local orthonormal frame we can write
\[
\spf{\xi}{\eta}=\fipd{\spbun{k}}{\xi}{\sigma(\fb{1})\eta}\fb{1}+\fipd{\spbun{k}}{\xi}{\sigma(\fb{2})\eta}\fb{2}.
\]

\begin{lem}
\label{covderivspf:lem}
Let $\spc$ be a spin$^c$ connection on $\spbun{k}$. 
If $\xi,\eta\in\Gamma(\spbun{k})$ and $X\in\Gamma(T\sphere)$ then
$\LCc[X]\spf{\xi}{\eta}=\spf{\spc[X]\xi}{\eta}+\spf{\xi}{\spc[X]\eta}$.
\end{lem}

\begin{proof}
We have $X\fipd{\sphere}{\spf{\xi}{\eta}}{\rho}
=\fipd{\sphere}{\LCc[X]\spf{\xi}{\eta}}{\rho}+\fipd{\sphere}{\spf{\xi}{\eta}}{\LCc[X]\rho}$ while
\begin{align*}
X\fipd{\spbun{k}}{\xi}{\sigma(\rho)\eta}
&=\fipd{\spbun{k}}{\xi}{\spc[X]\sigma(\rho)\eta}+\fipd{\spbun{k}}{\xi}{\spc[X](\sigma(\rho)\eta)}\\
&=\fipd{\spbun{k}}{\xi}{\spc[X]\sigma(\rho)\eta}+\fipd{\spbun{k}}{\xi}{\sigma(\rho)\spc[X]\sigma(\rho)\eta}
+\fipd{\spbun{k}}{\xi}{\sigma(\LCc[X]\rho)\eta}.
\end{align*}
The result now follows from the definition of $\spf{\xi}{\eta}$.
\end{proof}

Recall that Clifford multiplication extends naturally to $2$-forms; in particular $\sigma(\volS)=\sigma(\fb{1})\sigma(\fb{2})$ while for any $1$-form $\rho$ 
\begin{equation}
\label{sjsvcomm2:eq}
\sigma(\rho)\sigma(\volS)=-\sigma(*\rho).
\end{equation}

\begin{prop}
\label{spfdivcurl:prop}
Let $\Dirac{}$ be a Dirac operator on $\spbun{k}$.
If $\xi,\eta\in\Gamma(\spbun{k})$ then
\begin{equation}
\label{spfdiv:eq}
\delta\spf{\xi}{\eta}=i\fipd{\spbun{k}}{\Dirac{}\xi}{\eta}-i\fipd{\spbun{k}}{\xi}{\Dirac{}\eta}
\end{equation}
and
\begin{equation}
\label{spfcurl:eq}
d\spf{\xi}{\eta}=-i\bigl[\fipd{\spbun{k}}{\Dirac{}\xi}{\sigma(\volS)\eta}+\fipd{\spbun{k}}{\xi}{\sigma(\volS)\Dirac{}\eta}\bigr]\volS.
\end{equation}
\end{prop}

\begin{proof}
Let $\spc$ denote the spin$^c$ connection defining $\Dirac{}$. 
By \eqref{divconn:eq} and Lemma \ref{covderivspf:lem} 
\begin{align*}
\delta\spf{\xi}{\eta}
&=-\LCc[\vb{1}]\spf{\xi}{\eta}(\vb{1})-\LCc[\vb{2}]\spf{\xi}{\eta}(\vb{2})\\
&=-\spf{\spc[\vb{1}]\xi}{\eta}(\vb{1})-\spf{\spc[\vb{2}]\xi}{\eta}(\vb{2})
-\spf{\xi}{\spc[\vb{1}]\eta}(\vb{1})-\spf{\xi}{\spc[\vb{2}]\eta}(\vb{2})\\
&=-\bigipd{\bigl[\sigma(\fb{1})\spc[\vb{1}]+\sigma(\fb{2})\spc[\vb{2}]\bigr]\xi}{\eta}_{\spbun{k}}
-\bigipd{\xi}{\bigl[\sigma(\fb{1})\spc[\vb{1}]+\sigma(\fb{2})\spc[\vb{2}]\bigr]\eta}_{\spbun{k}}\\
&=-\ipd{i\Dirac{}\xi}{\eta}_{\spbun{k}}-\ipd{\xi}{i\Dirac{}\eta}_{\spbun{k}}.
\end{align*}
On the other hand working in a local orthonormal frame and applying \eqref{sjsvcomm2:eq} gives
\begin{align*}
*\spf{\xi}{\eta}
&=\ipd{\xi}{\sigma(-{*\fb{2}})\eta}_{\spbun{k}}\,{*\fb{1}}+\ipd{\xi}{\sigma(*\fb{1})\eta}_{\spbun{k}}\,{*\fb{2}}\\
&=-\ipd{\xi}{\sigma(\fb{2})\sigma(\volS)\eta}_{\spbun{k}}\fb{2}+\ipd{\xi}{\sigma(\fb{1})\sigma(\volS)\eta}_{\spbun{k}}(-\fb{1})
=\spf{\xi}{\sigma(\volS)\eta}.
\end{align*}
Together with \eqref{Diracsvolacomm:eq} and \eqref{spfdiv:eq} we get
\[
\delta{*\spf{\xi}{\eta}}
=-\delta\spf{\xi}{\sigma(\volS)\eta}
=i\ipd{\Dirac{}\xi}{\sigma(\volS)\eta}_{\spbun{k}}-i\ipd{\xi}{-\sigma(\volS)\Dirac{}\eta}_{\spbun{k}}.
\]
However $d=-*\delta*$ and $*1=\volS$ so \eqref{spfcurl:eq} follows.
\end{proof}

\begin{proof}[Proof of Proposition \ref{spinalign1:prop}]
Define a vector field $X'$ on $\sphere$ by $\alpha'=\fipd{\sphere}{X'}{\cdot}$.
Then $\fnorm{\sphere}{X'}=\fnorm{\sphere}{\alpha'}$ while
$\fipd{\spbun{k}}{\xi_1}{\sigma(\alpha')\xi_2}=\spf{\xi_1}{\xi_2}(X')$.
Hence
\begin{align*}
&\abs{\ipd{\xi_1}{\sigma(\alpha')\xi_2}}
\le\int_{\sphere}\bigabs{\fipd{\spbun{k}}{\xi_1}{\sigma(\alpha')\xi_2}}\,\volS
\le\int_{\sphere}\fnorm{\sphere}{X'}\,\fnorm{\sphere}{\spf{\xi_1}{\xi_2}}\,\volS\\
&\qquad{}\le\norm{\alpha'}_{L^\infty(\sphere)}\,\norm{\spf{\xi_1}{\xi_2}}_{L^1(\sphere)}
\le\frac{\pi}2\,\norm{\alpha'}_{L^\infty(\sphere)}\bigl[\norm{\delta\spf{\xi_1}{\xi_2}}_{L^1(\sphere)}+\norm{d\spf{\xi_1}{\xi_2}}_{L^1(\sphere)}\bigr]
\end{align*}
by Corollary \ref{1formL1est:cor}. 
On the other hand Proposition \ref{spfdivcurl:prop} leads to
\begin{align*}
\abs{\delta\spf{\xi_1}{\xi_2}}_{\sphere},\,\abs{d\spf{\xi_1}{\xi_2}}_{\sphere}
&\le\bigfnorm{\spbun{k}}{\sDirac{k}{t\alpha}\xi_1}\,\fnorm{\spbun{k}}{\xi_2}+\fnorm{\spbun{k}}{\xi_1}\,\bigfnorm{\spbun{k}}{\sDirac{k}{t\alpha}\xi_2}\\
&=(\abs{\lambda_1}+\abs{\lambda_2})\,\fnorm{\spbun{k}}{\xi_1}\,\fnorm{\spbun{k}}{\xi_2}
\end{align*}
(note that $\sigma(\volS)$ is a unitary operator in the fibres of $\spbun{k}$).
However
\[
2\int_{\sphere}\fnorm{\spbun{k}}{\xi_1}\,\fnorm{\spbun{k}}{\xi_2}\,\volS
\le\int_{\sphere}\bigl[\fnorm{\spbun{k}}{\xi_1}^2+\fnorm{\spbun{k}}{\xi_2}^2\bigr]\vtsp\volS
=\norm{\xi_1}_{L^2(\sphere)}^2+\norm{\xi_2}_{L^2(\sphere)}^2
=2.
\]
The result follows.
\end{proof}

\subsubsection*{Acknowledgements}
The author wishes to thank I.\ Sorrell and D.\ Vassiliev for several useful discussions. 
This research was supported by EPSRC under grant EP/E037410/1.
The author also acknowledges the hospitality of the Isaac Newton Institute for Mathematical Sciences in
Cambridge, where this work was completed during the programme Periodic and Ergodic Spectral
Problems.

\bigskip
\bigskip

\noindent
Daniel M.~Elton\\[6pt]
Department of Mathematics and Statistics\\
Fylde College\\
Lancaster University\\
Lancaster LA1 4YF\\
United Kingdom\\[6pt]
E-mail: d.m.elton@lancaster.ac.uk


\begin{thebibliography}{RS4}

\small

\bibitem[AMN1]{AMN1} C.~Adam, B.~Muratori and C.~Nash, \emph{Zero modes of the Dirac operator in three dimensions}, Phys.\ Rev.\ D, \textbf{60} (1999) 125001. 

\bibitem[AMN2]{AMN2} C.~Adam, B.~Muratori and C.~Nash, \emph{Degeneracy of zero modes of the Dirac operator in three dimensions}, Phys.\ Lett.\ B \textbf{485} (2000) 314--318.

\bibitem[AC]{AC} Y.~Aharonov and A.~Casher, \emph{Ground state of a spin-1/2 charged particle in a two-dimensional magnetic field}, Phys.\ Rev.\ A, \textbf{19}, \emph{no.~6} (1979) 2461--2462.

\bibitem[BE1]{BE1} A.~A.~Balinsky and W.~D.~Evans, \emph{On the zero modes of Pauli operators}, J.\ Funct.\ Anal.\ \textbf{179} (2001) 120--135.

\bibitem[BE2]{BE2} A.~A.~Balinsky and W.~D.~Evans, \emph{On the zero modes of Weyl-Dirac operators and their multiplicity}, Bull.\ London Math.\ Soc.\ \textbf{34} (2002) 236--242.


\bibitem[E1]{E1} D.~M.~Elton, \emph{New examples of zero modes}, J.\ Phys.\ A \textbf{33} (2000) 7297--7303. 

\bibitem[E2]{E2} D.~M.~Elton, \emph{The local structure of the set of zero mode producing magnetic potentials}, Commun.\ Math.\ Phys.\ \textbf{229} (2002) 121--139.

\bibitem[E3]{E3} D.~M.~Elton \emph{Approximate Zero Modes for the Pauli Operator on a Region}, to appear in J.\ Spectr.\ Theory. 

\bibitem[ET]{ET} D.~M.~Elton, N.~T.~Ta, \emph{Eigenvalue Counting Estimates for a Class of Linear Spectral Pencils with Applications to Zero Modes}, J.\ Math.\ Anal.\ Appl.\ \textbf{391}, (2012) 613--618.

\bibitem[ES]{ES} L.~Erd\H{o}s and J.~P.~Solovej, \emph{The kernel of Dirac operators on $\mathbb{S}^3$ and $\R^3$}, Rev.\ Math.\ Phys.\ \textbf{13} (2001) 1247--1280.

\bibitem[FS]{FS} W.~Freeden and M.~Schreiner, \emph{Spherical Functions of Mathematical Geosciences}, Springer-Verlag, Berlin (2009).

\bibitem[F]{F} T.~Friedrich, \emph{Dirac Operators in Riemannian Geometry}, Graduate Studies in Mathematics \textbf{25}, AMS, Providence (2000).

\bibitem[FLL]{FLL} J.~Fr\"ohlich, E.~Lieb and M.~Loss, \emph{Stability of Coulomb Systems with Magnetic Fields~I. The One Electron Atom} Commun.\ Math.\ Phys.\ \textbf{104} (1986) 251--270.  

\bibitem[GHL]{GHL} S.~Gallot, D.~Hulin and J.~Lafontaine, \emph{Riemannian Geometry}, 2nd Edition, Springer-Verlag, Berlin (1990).

\bibitem[H]{H} N.~Hitchin, \emph{Harmonic Spinors}, Advances in Math.\ \textbf{14} (1974) 1--55.

\bibitem[K]{K} T.~Kato, \emph{Perturbation Theory for Linear Operators}, 2nd Edition, Springer-Verlag, Berlin (1980).

\bibitem[LY]{LY} M.~Loss and H.~T.~Yau, \emph{Stability of Coulomb systems with magnetic fields~III. Zero energy states of the Pauli operator}, Commun.\ Math.\ Phys.\ \textbf{104} (1986) 283--290.




\bibitem[Ta]{T} N.~T.~Ta, \emph{Results on the Number of Zero Modes of the Weyl-Dirac Operator}, PhD Thesis, Lancaster University (2009).

\bibitem[T]{Th} B.~Thaller, \emph{The Dirac Equation}, Springer-Verlag, Berlin (1992).

\bibitem[W]{W} H.~Weyl, \emph{Inequalities between two kinds of eigenvalues of a linear transformation}, Proc.\ Nat.\ Acad.\ Sci.\ U.S.A.\ \textbf{35} (1949) 408--411.

\end{thebibliography}
\end{document}